\documentclass[acmsmall,screen,letterpaper]{acmart}
\pdfoutput=1

\usepackage{mathtools}

\usepackage[capitalise]{cleveref}

\usepackage{todonotes}
\usepackage{apxproof}

\definecolor{darkgreen}{RGB}{0, 128, 0}

\usepackage{listings}
\usepackage{wrapfig}

\newcommand{\hurl}[1]{{\scriptsize\UrlFont\href{http://#1}{\path{#1}}}}
\newcommand{\hsurl}[1]{{\scriptsize\UrlFont\href{https://#1}{\path{#1}}}}

\newcommand\N{\mathbb{N}}

\newcommand{\false}{\bot}

\newcommand\init{\mathrm{init}}

\renewcommand\P{\mathcal{P}}
\newcommand\R{\mathcal{R}}

\newcommand\Stmt{\mathbf{Stmt}}
\newcommand\st{\mathit{s\!t}}

\newcommand\keyword[1]{\texttt{\color{blue}\bfseries #1}}
\newcommand\assume[1]{\keyword{assume}\,\texttt{#1}}
\newcommand\assert[1]{\keyword{assert}\,\texttt{#1}}
\newcommand\havoc[1]{\keyword{havoc}\,\texttt{#1}}
\newcommand\syncStmt[2]{\texttt{\keyword{for}\,$j\neq i$\,:\,#1\,:=\,#2}}

\newcommand\Loc{\mathbf{Loc}}
\newcommand\pc[1]{\mathit{pc_{#1}}}

\newcommand\locals{\mathbf{Var}_\mathsf{local}}
\newcommand\globals{\mathbf{Var}_\mathsf{global}}

\newcommand\inst[2]{{#1\!\!:\!#2}}

\newcommand\enabled[1]{\mathit{enabled}({#1})}

\newcommand\sem[1]{\llbracket #1 \rrbracket}

\colorlet{th1}{yellow}
\colorlet{th2}{green!80!blue}
\colorlet{th3}{blue!30!white}
\tikzset{stmt/.style={font=\ttfamily\scriptsize,shape=rectangle,rounded corners=.3em,fill=gray!40,inner xsep=.3em,inner ysep=0em}}
\newcommand{\stsmcol}[2]{\raisebox{1mm}{\tikz[baseline]{\node[stmt,fill=#2] at (0,0){\raisebox{-1mm}[2mm][2mm]{\hspace{.0em}\texttt{\strut\small #1}\hspace{.0em}\strut}};}}}

\newcommand\comm[1][]{\mathrel{\substack{\rotatebox{0}{$\curvearrowright$}\\[-.5em]\rotatebox{180}{\hspace{0.3em}$\curvearrowright$}}_{#1}}}
\newcommand\notcomm[1][]{\mathrel{\comm[#1]\hspace{-1.3em}{\raisebox{-.3em}{\scalebox{1.75}{/}}}\hspace{.3em}}}
\newcommand\semicomm[1][]{\mathrel{\curvearrowright_{#1}}}

\newcommand\id[1]{{\mathit{id}_{#1}}}
\newcommand\sleep[1]{{\mathit{sleep}_{#1}}}

\newcommand\instrument[1]{\iota(#1)}

\newcommand\sleepInstr[1]{#1_\sleep{}}

\newcommand\prefTest[2]{\mathrm{pref}(#1, #2)}

\newcommand\chcTm[2]{\mathbb{TM}(#1,#2)}
\newcommand\chcSymb[2]{\chcTm{\sleepInstr{#1}}{#2}}
\newcommand\chcExpl[2]{\mathbb{TM}_{\sleep{}}(#1,#2)}

\newcommand\search[1]{\mathbf{Search}(#1)}
\newcommand\sol[1]{\mathbf{Sol}(#1)}

\newcommand\Inv{\mathit{Inv}}

\newcommand\prefCond[3]{\widetilde{\mathsf{pref}}(#1, #2 / #3)}

\newcommand\commCond[3]{\mathsf{comm}({#1},\inst{#3}{#2})}
\newcommand\commCondSt[3]{\mathsf{comm}({#1},\stsmcol{\inst{#3\,}{$#2$}}{th1})}

\crefname{observation}{Observation}{Observation}

\begin{document}

\newtheorem{observation}[theorem]{Observation}

\title{Commutativity Simplifies Proofs of Parameterized Programs}

\author{Azadeh Farzan}
\orcid{0000-0001-9005-2653}             %
\affiliation{
  \institution{University of Toronto}            %
  \city{Toronto}
  \country{Canada}                    %
}
\email{azadeh@cs.toronto.edu}

\author{Dominik Klumpp}
\orcid{0000-0003-4885-0728}             %
\affiliation{
  \institution{University of Freiburg}            %
  \city{Freiburg im Breisgau}
  \country{Germany}                    %
}
\email{klumpp@informatik.uni-freiburg.de}          %

\author{Andreas Podelski}
\orcid{0000-0003-2540-9489}             %
\affiliation{
  \institution{University of Freiburg}           %
  \city{Freiburg im Breisgau}
  \country{Germany}                   %
}
\email{podelski@informatik.uni-freiburg.de}

\begin{abstract}
  \emph{Commutativity} has proven to be a powerful tool in reasoning about concurrent programs.
  Recent work has shown that a commutativity-based \emph{reduction} of a program may admit simpler proofs than the program itself.
  The framework of lexicographical program reductions was introduced to formalize a broad class of reductions which accommodate sequential (thread-local) reasoning as well as synchronous programs.
  Approaches based on this framework, however, were fundamentally limited to program models with a {\em fixed/bounded} number of threads.
  In this paper, we show that it is possible to define an effective parametric family of program reductions that can be used to find simple proofs for {\em parameterized programs}, i.e., for programs with an unbounded number of threads.
  We show that reductions are indeed useful for the simplification of proofs of parameterized programs, in a sense that can be made precise:
  A  reduction of a parameterized program may admit a proof which uses {\em fewer} or {\em less sophisticated ghost variables}.
  The reduction may therefore be within reach of an automated verification technique,
  even when the original parameterized program is not.
  As our first technical contribution, we introduce a notion of reductions for parameterized programs such that
  the reduction $\R$ of a parameterized program $\P$ is again a parameterized program (the thread template of $\R$ is obtained by source-to-source transformation of the  thread template of $\P$).
  Consequently, existing techniques for the  verification of parameterized programs can be directly applied to $\R$ instead of $\P$.
  Our second technical contribution is that we define an appropriate family of {\em pairwise preference orders} which can be effectively used as a parameter to produce different lexicographical reductions.
  To determine whether this theoretical foundation amounts to a usable solution in practice, we have implemented the approach, based on a recently proposed framework for parameterized program verification.
  The results of our preliminary experiments on a representative set of examples are encouraging.
\end{abstract}



\maketitle

\section{Introduction}
\label{sec:intro}

The framework of {\em trace theory} (formulated by Mazurkiewicz in 1987) formalizes equivalence relations for concurrent program runs based on a commutativity relation over the set of atomic steps taken by individual program threads. Two program statements of different threads \emph{commute} if the order in which we execute them is irrelevant to the outcome of the execution. Two program runs are {\em equivalent} up to commutativity if one can be acquired from another through successive swaps of adjacent commutative program steps.
For any program $A$, we call a program $B$ a {\em reduction} of $A$ if and only if $B$ includes at least one representative from each (commutativity) equivalence class of behaviours in $A$.
Recent work \cite{cav19:hypersafety,popl20:red-safety,pldi22:sound-seq,lics2023} has shown that some reductions of a program admit {\em simpler} proofs than the program itself.
More specific versions of this observation had already been made in the literature of concurrent and distributed program verification.
In particular, it is exploited in the context of verification of distributed programs by favouring the verification of synchronous (or almost synchronous) programs in place of asynchronous programs with the rationale that the synchronous program admits a simpler proof \cite{kragl:layered,gleissenthal:pretend-synchrony,Genest07}.

The common thread in all these contexts is that there is often a lot of redundancy in the set of behaviours of a concurrent program,
and removing redundant behaviours with complicated proofs in favour of those with simpler proofs simplifies the entire reasoning task. The choice of a program {\em reduction}, then, is a choice of which {\em representatives} from equivalence classes of program behaviours stay and which ones go.
Traditionally, people have opted for canonical choices: those that maximize sequential (local) reasoning in the case of concurrent programs \cite{elmas:calculus-atomic,kragl:layered}, or those that get as close as possible to a {\em synchronous} program \cite{gleissenthal:pretend-synchrony,Genest07} for distributed protocols.
As such, each such framework makes an a priori  assumption about a particular type of reduction.
In recent work, however, a family of {\em parametric lexicographical program reductions} \cite{cav19:hypersafety,popl20:red-safety,pldi22:sound-seq} were introduced that formalized a broad (infinite) class of reductions that would include both canonical choices.
The idea is that different program verification tasks may respond best to different strategies for picking representatives.
By taking a {\em lexicographic order} as a parameter to a reduction that chooses the (lexicographically) {\em least} representative of each equivalence class, one controls the composition of the reduction.

These frameworks, however, were fundamentally built based on an assumption that the alphabet of program actions is {\em finite}, and therefore, they can only be applied to program models with a {\em fixed/bounded} number of threads.
This brings us to the central research question in this paper:
``For programs with unboundedly many threads, is it possible to define an effective parametric family of program reductions that can be exploited for finding simple proofs?'' \
This paper presents an affirmative answer to this question for {\em parameterized concurrent programs}.
A \emph{parameterized program} $\P$ stands for an infinite family of programs $\P(n)_{n\in\N}$.
Each program $\P(n)$ arises from taking a number $n$ of threads, where $n$ is not bounded.
Each thread runs an instance of the same given thread template. This is without the loss of generality, since  well-known encoding tricks~\cite{popl17:thread-modular} accommodate the use of multiple thread templates.

It is well-understood, even outside the realm of algorithmic verification,  that modular reasoning techniques for parameterized programs (e.g. Owicki-Gries for parameterized programs ~\cite{prensa-nieto:owicki-complete})
are only {\em complete} in the presence of the full power of {\em history} variables.
Therefore, program proofs may require highly nontrivial ghost variables, which are notoriously hard to compute and reason about automatically. In contrast, in the fixed thread case, the canonical choice of {\em program counters} is always available and mainly becomes a time/complexity issue for verification algorithms.
This paper argues that reductions can help simplify proofs of parameterized concurrent programs, in a sense that can be made precise based on the ghost variables required for the proof. We make
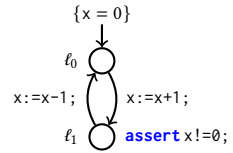
\begin{wrapfigure}[7]{r}{0.25\textwidth}
 \vspace{-3.5mm}
  \begin{tikzpicture}[thick,font=\scriptsize]
    \node[draw,circle,label={[]left:$\ell_0$}] (0) {};
    \node[draw,circle,label={[]left:$\ell_1$},label={right:\assert{x!=0;}},below of=0] (1) {};
    \node[above=3mm of 0,inner sep=0] (pre) {$\{\texttt{x}=0\}$};
    \draw[<-] (0) -- (pre);
    \draw[->] (0) edge[bend left] node[auto]{\texttt{x:=x+1;}} (1);
    \draw[->] (1) edge[bend left] node[auto]{\texttt{x:=x-1;}} (0);
  \end{tikzpicture}
  \caption{Template for $\P^{\pm}$}
  \label{fig:inc-dec-template}
\end{wrapfigure}
the observation that a reduction of a parameterized program may admit a proof which uses {\em fewer or}
{\em less sophisticated ghost variables}
and may therefore have a higher chance of being within the reach of an automated verification technique.

As a simple example to make this observation concrete,
consider the the parameterized program
$\P^{\pm}$, given by the thread template in \cref{fig:inc-dec-template}.
The goal is to prove the property that whenever a thread is in location~$\ell_1$, the global variable \texttt{x} is non-zero,
assuming \texttt{x} is initially $0$.
It can be shown
there does not exist a proof (formally, a proof in the form of an \emph{Ashcroft invariant})
if one does not introduce a ghost variable~\cite{popl17:thread-modular}.
Intuitively, the proof needs to keep track of the number of threads that have  already executed their increment but not yet the matching decrement.
Now consider the reduction where the threads are executed sequentially one after the other (sequential composition).
The reduction is sound because all statements of two different threads commute
(since we do not model the specification \assert{x!=0} as a statement, we are not concerned with its commutativity).
The proof for the reduction does not need any ghost variables.
We will use the example later as a running example (see \cref{sec:param-red}).

Our first technical contribution is a notion of a reduction for a parameterized program.
The reduction $\R$ of the parameterized program $\P$ can be viewed as {\em a family of lexicographical reductions}.
This means that $\R$ stands for an infinite family of programs $\R(n)_{n\in\N}$ where for each $n$, $\R(n)$ is a (lexicographical) reduction of $\P(n)$.
Crucially, the infinite family can be {\em finitely represented}.
In fact, the reduction $\R$ is again a parameterized program,
and the thread template of $\R$ is obtained by source-to-source transformation of the  thread template of $\P$.
The key benefit of this observation is that existing techniques for verification of parameterized programs can now be directly applied to~$\R$ instead of $\P$.

Reductions that favour program behaviours with long sequential blocks, like the sequential composition for the example in \cref{fig:inc-dec-template},
can be generated using lexicographical reductions based on thread orders; i.e. when statements of each thread are grouped together and ordered wrt.\ statements of other threads according to their thread identifiers.
In \cref{sec:motivating-example}, we present an example that demonstrates why, in the context of parameterized program verification, other reductions like {\em lockstep} reductions, may be essential if proof simplification is the desired outcome.

Our second technical contribution is that we define an appropriate family of orders, called {\em pairwise preference orders}, that can be effectively used as a parameter to produce many different lexicographical reductions of the same program given the same commutativity relation (including the above-mentioned lockstep reduction). This generalizes similar results from the literature on how reductions for a fixed number threads are generated parametric on order relations \cite{cav19:hypersafety,popl20:red-safety,pldi22:sound-seq}.
We show that, as in the case of thread orders, reductions of a parameterized program $\P$ parametric on pairwise preference orders can also be finitely represented as parameterized program $\R$, with the same correspondence between $\P(n)$ and $\R(n)$ for all $n$.

The two technical contributions outlined so far put forward an algorithmic path for verifying parameterized concurrent programs using a broad family of reductions.  To determine whether this amounts to a usable solution in practice, we selected the proof method based on \emph{thread-modular proofs at many levels}~\cite{popl17:thread-modular}
to instantiate and evaluate this solution.
The proof method encodes the existence of a proof of a specific form (an Ashcroft invariant with a number $k$ of universal quantifiers over thread IDs) for an input parameterized program $\P$ as a satisfiability problem of a set of constraints in a specific form (CHC, for Constrained Horn Clauses).
To use the proof method for verifying a reduction of the input parameterized program, we apply the proof method to our proposed parameterized reduction, i.e., to the parameterized program $\R$.

We implemented the construction of the parameterized program $\R$
and the constraint generation according to~\citet{popl17:thread-modular}.
We evaluated the approach on a set of 19 parameterized programs taken from the literature,
by discharging the generated constraints with several off-the-shelf CHC solvers.
The results are very encouraging: The implementation succeeded in verifying the reductions of 14 programs, only 4 of which can be verified without the use of reductions.

It is noteworthy that our proposal for parameterized reductions (and therefore, the corresponding set of CHC constraints) have the desired property that any Ashcroft invariant of the original program is also a valid invariant for the reduced program.
The converse does not hold; i.e., the reduction~$\R$ may admit an Ashcroft invariant that is not a valid invariant of the original program $\P$, and a proof in the form of an Ashcroft invariant may not exist for $\P$ even though it does for $\R$.

The property of the \emph{conservative extension} of the validity of an Ashcroft invariant from $\R$ to $\P$ does not, however, mean that we are (in practice) able to compute a proof in the form of an Ashcroft invariant for $\R$ whenever we are able to compute one for $\P$.
In fact, the parameterized program uses a set of additional variables as the means of encoding the reduction.
It is thus natural to wonder whether the task of the CHC solver could somehow become harder because it has to deal with constraints over a larger set of variables, and, if so, whether anything can be done to alleviate this issue. We investigate this question systematically in \cref{sec:red-proofs} and propose an alternative encoding with fewer variables. This new encoding is an orthogonal contribution of this paper.  It is inspired by  the idea of {\em symmetry reduction}~\cite{clarke1998symmetry}.
Intuitively, in the encoding based on \citet{popl17:thread-modular}, the solver is forced to prove the correctness of symmetry-equivalent classes of reductions. In \cref{sec:break-symm} we demonstrate how the CHC encoding can be modified so that this redundancy is eliminated.

To conclude, this paper proposes a way of incorporating commutativity-based reductions into, in principle, any existing parameterized verification methodology.  In particular, it makes the following contributions:
\begin{itemize}
\item We observe that reductions simplify proofs of parameterized programs in a precise sense:
Proofs of reductions require less complex ghost state than the proofs of original programs; this can manifest as the need for less complicated information to be recorded in ghost variables, or that simply fewer ghost variables are needed overall (\cref{sec:motivating-example}).
\item The theoretical formulation of a parameterized reduction in two parts:
\begin{enumerate}
\item  We formulate a lexicographical reduction of a parameterized program and show that it can be finitely represented, namely again as a parameterized program  (\cref{sec:param-red}).
\item  We propose an appropriate notion of preference orders for the parameterized context and show that the construction of a lexicographical reduction from a parameterized program can be made parametric on the preference order  (\cref{sec:pref-orders}).
\end{enumerate}
\item We give an improved formulation of the search problem for an Ashcroft invariant, by breaking some inherent but redundant symmetries in the search space and the corresponding solution space without affecting soundness or completeness of the methodology (\cref{sec:break-symm}).
\end{itemize}

\section{Motivating Example}
\label{sec:motivating-example}
We demonstrate the benefits of commutativity for proof simplification
using the parameterized program $\P^\mathrm{notify}$ shown in \cref{fig:notify}.
This program models a distributed system, in which one thread (called \texttt{notifier})
generates data through some computation (line 6-9),
and broadcasts it to an unbounded number of \texttt{listener} threads (line 11-13).
The threads communicate via a message queue,
which is here modeled via an infinite \texttt{queue} array
along with an integer \texttt{current} pointing to the head of the queue
(specifically, to the first invalid entry).

\begin{figure}
\begin{minipage}[t]{0.3\textwidth}
\begin{lstlisting}[numbers=left]
notifier() {

  last := 0;

  while (true) {
    // generate data
    havoc data;
    assume data > last;
    last := data;

    // send new data
    queue[current] := data;
    current := current + 1;
  }
}
\end{lstlisting}
\end{minipage}
\qquad
\begin{minipage}[t]{0.3\textwidth}
\begin{lstlisting}[numbers=left,firstnumber=16]
listener() {
  idx := current;
  prev := 0;

  while (true) {
    // receive data
    assume idx < current;
    msg := queue[idx];
    idx := idx + 1;

    // check data
    assert prev < msg;
    prev := msg;
  }
}
\end{lstlisting}
\end{minipage}%
\caption{
  The program $\P^\mathrm{notify}$.
  The variables \texttt{current} (an integer) and \texttt{queue} (an integer array) are global, all other variables are local.
  An instance $\P^\mathrm{notify}(n)$ consists of a single \texttt{notifier} thread and $n$ \texttt{listener} threads.
}%
\label{fig:notify}%
\end{figure}%

Each \texttt{listener} thread joins the conversation by setting its thread-local \texttt{idx} variable to the value of \texttt{current}.
The \texttt{listener} then continuously waits for new data to appear in the queue (line~22).
When data has arrived, it reads the message from the queue (line~23-24).
In the next step, the \texttt{listener} checks the integrity of the received message.
In particular, it checks that the received value is greater than the previous message (line~26-28).

Showing correctness of this program is non-trivial;
even with ghost variables, a proof is challenging.
An unbounded amount of time may pass between the moment when a message is sent by the \texttt{notifier} thread,
and when the last \texttt{listener} receives it.
Thus, for certain traces, one must keep track of the \texttt{idx} variables of unboundedly many listener threads,
not just a finite subset of them.

There exists however a subset of traces,
for which the correctness argument is much simpler.
Namely, consider those traces where every message sent by the \texttt{notifier}
is immediately received and checked by all \texttt{listeners} that have already joined the conversation
(i.e., all \texttt{listeners} that will \emph{ever} receive the message).
Let us call these traces \emph{synchronous}.
In synchronous traces, the difficulty of reasoning about an unbounded number of messages already sent
but not yet received by some \texttt{listener}
completely disappears.
At any point, there is at most one such message,
and consequently, the proof has to reason only about one message.

Of course, synchronous traces make up only a small fragment of the many interleavings of the program.
To show correctness of the program, we must establish that every trace is correct.
Here, commutativity comes to the rescue:
We observe that for many statements of the program,
the order in which they are executed does not affect the outcome.
We say that such statements \emph{commute} with each other.
We exploit this observation by repeatedly swapping commuting statements,
and thereby reorder any arbitrary trace of the program to an \emph{equivalent} synchronous trace.
Through a meta-argument (i.e., the soundness theorem of our approach),
we establish that any trace that is equivalent to a correct synchronous trace must itself be correct.
Thus, it suffices for a proof to show correctness of synchronous traces,
in order to conclude that the program is correct.

Consider for instance the statements \stsmcol{\texttt{last:=data}}{th1} (line~9) and \stsmcol{\texttt{prev:=msg}}{th2} (line~28).
Executing these statements in either order yields the same result,
i.e., the statements commute with each other.
Similarly, we can argue that all statements of the \texttt{notifier} thread commute with the statement \stsmcol{\texttt{prev:=msg}}{th2}.
Therefore, we consider for instance the following traces to be equivalent:
\begin{align*}
  \stsmcol{\havoc{data}}{th1}\, \stsmcol{\texttt{data\,>\,last}}{th1}\, \stsmcol{\texttt{last:=data}}{th1}\, \stsmcol{\texttt{queue[current]:=data}}{th1}\, \stsmcol{\texttt{current:=current+1}}{th1}\, \stsmcol{\texttt{prev:=msg}}{th2}\\
  \sim \stsmcol{\havoc{data}}{th1}\, \stsmcol{\texttt{data\,>\,last}}{th1}\, \stsmcol{\texttt{last:=data}}{th1}\, \stsmcol{\texttt{queue[current]:=data}}{th1}\, \stsmcol{\texttt{prev:=msg}}{th2}\, \stsmcol{\texttt{current:=current+1}}{th1}\\
  \sim \stsmcol{\havoc{data}}{th1}\, \stsmcol{\texttt{data\,>\,last}}{th1}\, \stsmcol{\texttt{last:=data}}{th1}\, \stsmcol{\texttt{prev:=msg}}{th2}\, \stsmcol{\texttt{queue[current]:=data}}{th1}\, \stsmcol{\texttt{current:=current+1}}{th1}\\
  \sim \stsmcol{\havoc{data}}{th1}\, \stsmcol{\texttt{data\,>\,last}}{th1}\, \stsmcol{\texttt{prev:=msg}}{th2}\, \stsmcol{\texttt{last:=data}}{th1}\, \stsmcol{\texttt{queue[current]:=data}}{th1}\, \stsmcol{\texttt{current:=current+1}}{th1}\\
  \sim \stsmcol{\havoc{data}}{th1}\, \stsmcol{\texttt{prev:=msg}}{th2}\, \stsmcol{\texttt{data\,>\,last}}{th1}\, \stsmcol{\texttt{last:=data}}{th1}\, \stsmcol{\texttt{queue[current]:=data}}{th1}\, \stsmcol{\texttt{current:=current+1}}{th1}\\
  \sim \stsmcol{\texttt{prev:=msg}}{th2}\, \stsmcol{\havoc{data}}{th1}\, \stsmcol{\texttt{data\,>\,last}}{th1}\, \stsmcol{\texttt{last:=data}}{th1}\, \stsmcol{\texttt{queue[current]:=data}}{th1}\, \stsmcol{\texttt{current:=current+1}}{th1}
\end{align*}
These equivalences allows us reorder entire iterations of the \texttt{notifier} thread, i.e., the computation and broadcast of new data,
wrt.\ the statement \stsmcol{\texttt{prev:=msg}}{th2}.
We proceed similarly with respect to the other statements of the \texttt{listener} thread,
as well as for the statements of two different \texttt{listener} threads.

For some of these other statements, we must consider broader notions of commutativity.
As an example,
we cannot generally claim that the order in which the statements \stsmcol{\texttt{queue[current]:=data}}{th1} and \stsmcol{\texttt{msg:=queue[idx]}}{th2} are executed
does not affect the outcome.
Specifically, if we have $\texttt{current} = \texttt{idx}$, the order is in fact crucial.
However, observe that the program ensures that,
whenever the statement \stsmcol{\texttt{msg:=queue[idx]}}{th2} is executed,
it actually holds that $\texttt{idx} < \texttt{current}$.
In such contexts, the order in which the statements are executed is indeed irrelevant.
Hence we can say that the statements commute \emph{within this particular program}.

The essential insight of commutativity reasoning is this:
It suffices for a proof to cover a so-called \emph{reduction} of a program,
i.e., a subset of traces such that each program trace is equivalent to a trace in the reduction.
In our example, the reduction is formed by the set of synchronous traces.
By soundness of commutativity,
we can conclude that,
if the reduction is proven correct, the entire program must be correct.
In this manner, our approach can verify the program $\P^\mathrm{notify}$
by giving a proof for synchronous traces.
As discussed, a proof for the set of synchronous traces is much simpler than a proof for all traces,
as it does not require complex ghost state or quantified invariants.

As another example where commutativity simplifies the proof,
let us consider the program~$\P^{K\pm}$, with the thread template shown in \cref{fig:hierarchy-collapse}.
The program has a global variable \texttt{x}, which is initially $0$.
The program uses a constant $K$ for which we assume a fixed value.
Each thread repeatedly checks if the
\begin{wrapfigure}[8]{r}{0.25\textwidth}
  \vspace{-4mm}
  \begin{tikzpicture}[thick,font=\scriptsize]
    \node[draw,circle,label={[]left:$\ell_0$}] (0) {};
    \node[draw,circle,label={[]left:$\ell_1$},label={right:\assert{x!=0;}},below=1cm of 0] (1) {};
    \node[above=3mm of 0,inner sep=0] (pre) {$\{\texttt{x}=0\}$};
    \draw[<-] (0) -- (pre);
    \draw[->] (0) edge[bend left] node[auto,align=left]{\assume{x < K}\\\texttt{x:=x+1}} (1);
    \draw[->] (1) edge[bend left] node[auto]{\texttt{x:=x-1}} (0);
  \end{tikzpicture}
  \caption{Template for $\P^{K\pm}$}
  \label{fig:hierarchy-collapse}
\end{wrapfigure}
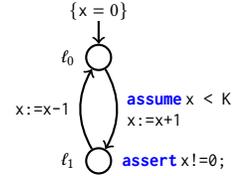
current value of \texttt{x} is less than $K$,
and if so, increments \texttt{x}.
It asserts that \texttt{x} is non-zero,
eventually decrements \texttt{x} again,
and begins the loop anew.

This program is similar to the example discussed in the introduction,
yet due to the guard using the constant $K$, the proof is in some sense simpler:
The value of a ghost variable counting the number of threads in location $\ell_1$
can never exceed $K$.
Thus, we can alternatively consider the local state of $K$ other threads as ghost state.
Specifically, if a thread is in location $\ell_1$, and some number $m$ (with $0\leq m \leq K$) of the $K$ other threads are also in location $\ell_1$,
we know that $\texttt{x} \geq m+1$,
and therefore, decrementing \texttt{x} does not violate the assert statement in any thread:
Either we have $m > 0$, in which case \texttt{x} is still positive after the decrement,
or $m = 0$, in which case none of the threads is in location $\ell_1$.

It has been shown that for any value of $K$,
a proof does indeed need to consider at least $K$ additional threads as ghost state (and thus overall consider $K+1$ threads at a time)
in order to show correctness of this program~\cite{popl17:thread-modular}.
However, commutativity simplifies the required ghost state.

Let us investigate the commutativity in $\P^{K\pm}$.
Two statements \stsmcol{x:=x-1}{th1} and \stsmcol{x:=x-1}{th2} of different threads commute,
as do two statements \stsmcol{\assume{x\,<\,K};\,x:=x+1}{th1} and \stsmcol{\assume{x\,<\,K};\,x:=x+1}{th2} of different threads.
For the statements \stsmcol{x:=x-1}{th1} and \stsmcol{\assume{x\,<\,K};\,x:=x+1}{th2},
the order of execution may indeed matter.
But whenever it is possible to execute the sequence \stsmcol{\assume{x\,<\,K};\,x:=x+1}{th2}\,\stsmcol{x:=x-1}{th1},
it is also possible to execute the sequence \stsmcol{x:=x-1}{th1}\,\stsmcol{\assume{x\,<\,K};\,x:=x+1}{th2}
with the same effect (\texttt{x} is not modified),
i.e., the latter sequence allows a strict superset of executions.
Thus we can verify traces containing the sequence \stsmcol{x:=x-1}{th1}\,\stsmcol{\assume{x\,<\,K};\,x:=x+1}{th2}
and conclude that traces containing the sequence \stsmcol{\assume{x\,<\,K};\,x:=x+1}{th2}\,\stsmcol{x:=x-1}{th1} are also correct.

Analogously to the example in \cref{fig:inc-dec-template},
we exploit this commutativity (or \emph{semi-commutativity}) to reorder any trace of the program such that all statements of a thread are executed in a single block.
For the resulting reduction of the program,
it is sufficient to consider the local state of a single additional thread as ghost state, rather than $K$ threads.
If a thread is in $\ell_1$, and the other thread (which serves as ghost state) is also in $\ell_1$, we know that $\texttt{x}\geq 2$,
so \texttt{x} remains positive after a decrement.
If the ``ghost thread'' is not in $\ell_1$, neither thread executes the assert statement.
Commutativity has again simplified the ghost state required to prove correctness of the program.

\section{Parameterized Concurrent Programs}

A parameterized program $\P$ is given by its thread template (a control flow graph) and a set of thread-local variables,
i.e., $\P = \langle \Loc, \Delta, \ell_\init, \locals \rangle$
with a finite set of locations $\Loc$,
a finite transition relation $\Delta \subseteq \Loc\times\Stmt\times\Loc$ (where $\Stmt$ is the set of atomic program statements),
an initial location $\ell_\init \in \Loc$,
and a set of thread-local variables $\locals$.
Any variable not in $\locals$ is considered global.
We denote the set of global variables as $\globals$.

The enabled statements $\enabled{\ell}$ of a location $\ell$ are the statements $\st$ such that $\langle \ell,\st,\ell' \rangle\in\Delta$ for some $\ell'$.
We assume that the only case in which $\enabled{\ell}$ contains more than one statement is the case of a branch (or loop head),
and thus $\enabled{\ell} = \{\assume{$e$}, \assume{$\lnot e$} \}$ for some branching condition (or loop guard) $e$.
This assumption is only required for the minimality of our reduction (\cref{prop:minimal});
the soundness of our approach does not rely on it.

A parameterized program describes a family of programs.
For each number of threads $n\in\N$, the instance of the program with $n$ threads is denoted by $\P(n)$.
The variables of the program instance $\P(n)$ consist of the global variables, as well as indexed local variables $x_i$ for each $i\in\{1,\ldots,n\}$ and $x\in\locals$.
The program instance $\P(n)$ uses \emph{indexed statements} $\inst{\st}{i}$,
where $\st\in\Stmt$ is a statement as it appears in the thread template,
and the thread index $i\in\{1,\ldots,n\}$ indicates which thread executes the statement.

\paragraph{Traces}
A thread template defines a languages $L$ over the alphabet $\Stmt$,
consisting of all sequences of statements that label any path from the initial location (regardless which location is reached in the end).

The language of an instance $\P(n)$ of the parameterized program $\P$
is a language of \emph{traces}, i.e., sequences of \emph{indexed} statements. %
For the language $L$ defined by the thread template of $\P$,
let $L[i]$ be the language $L$ where every statement $\st$ has been replaced by the indexed statement $\inst{\st}{i}$.
The program instance $\P(n)$ then defines the language of all traces allowed by the control flow of $\P$:
\begin{align*}
  \P(n) &= L[1] \parallel \ldots \parallel L[n], %
\end{align*}
where $\parallel$ denotes the shuffle operation on languages.

\paragraph{Semantics}
We assume that each statement $\st\in\Stmt$ has an associated semantics $\sem{\st}$,
given by a binary input/output relation between valuations of the program variables.
In particular, the semantics of assignment statements $x\texttt{:=}e$ and assume statements \assume{$e$} is as one would expect.

We extend this semantics to indexed statements.
Executing the indexed statement $\inst{\st}{i}$ may modify the global variables as well as the indexed local variables $x_i$,
but leaves local variables of other threads unmodified.
Formally, we define the semantics of an indexed statement as follows:
\[
  \sem{\inst{\st}{i}} := \big\{\,(s_1,s_2)\mid (s_1|_i,s_2|_i)\in\sem{\st} \land \forall x\in\locals\,.\,\forall j\neq i\,.\,s_2(x_j)=s_1(x_j)\,\big\},
\]
where $s_1,s_2$ are valuations of the variables of $\P(n)$, and $s|_i$ is the unique valuation of the program variables such that $s|_i(x)=s(x_i)$ for local variables~$x$ and $s|_i(g) = s(g)$ for global variables $g$.

Based on these semantics of atomic statements, we define the semantics of each program instance.
A \emph{configuration} of $\P(n)$ is a pair $\langle \vec\ell, s \rangle$, where $\vec\ell = \langle \ell_1,\ldots,\ell_n \rangle \in \Loc^n$ denotes the control locations of the running threads,
and $s$ is a valuation of the variables of the program instance $\P(n)$.
We say that the configuration $\langle\vec\ell,s\rangle$ is \emph{initial} if $\vec\ell=\langle\ell_\init,\ldots,\ell_\init\rangle$.%

Let $\langle\vec\ell,s\rangle$ be a configuration, such that $\langle\ell_i,\st,\ell_i'\rangle \in \Delta$ is a transition of the thread template,
and such that there is a successor valuation $s'$ with $(s,s')\in\sem{\inst{\st}{i}}$.
From this configuration, the program can execute $\inst{\st}{i}$.
Thread $i$ moves to control location $\ell_i'$, whereas all other threads remain at the same location ($\ell_j' = \ell_j$ for all $j\neq i$).
We write $\langle \vec\ell,s \rangle \xrightarrow{\inst{\st\;}{\,i}} \langle \vec\ell',s' \rangle$.
A trace $\tau = \inst{\st_1}{i_1}\ \ldots \ \inst{\st_m}{i_m}$ is \emph{feasible}
if there exists a corresponding sequence of configurations (called an \emph{execution}) $\langle \vec\ell^{(1)}, s_1\rangle \xrightarrow{\inst{\st_1\;}{\;i_1}} \ldots \xrightarrow{\inst{\st_n\;}{\;i_n}}\langle \vec\ell^{(m)}, s_m\rangle$,
and $\langle\vec\ell^{(1)},s_1\rangle$ is initial.
If a trace is not feasible, it is \emph{infeasible}.

\paragraph{Synchronous Statements}
Our approach uses a particular kind of statements, so-called \emph{synchronous statements}~\cite{popl17:thread-modular}.
A thread can execute a synchronous statement to (atomically) update the local variables for all (unboundedly many) other threads.
Synchronous statements have the form
\[
  \syncStmt{$x_j$}{$e$}
\]
where $j$ and $i$ are symbolic indices representing the thread whose variables are updated ($j$) and the thread that executes the statement ($i$).
The updated variable $x$ must be a local variable $(x\in\locals$).
The expression $e$ may refer to global variables, as well as local variables $y_i, y_j$ indexed by $i$ or $j$.
Additionally, we allow $e$ to refer to special variables $\pc{i}$ and $\pc{j}$, which represent the current control locations of thread $i$ resp.\ $j$.

\paragraph{Correctness and Proofs}
A specification for a parameterized program $\P$ consists of a precondition $\mathit{pre}$,
and a partial map $\mathit{assert}$ from program locations to formulae over the program variables.
Both the precondition $\mathit{pre}$ and an assertion $\mathit{assert}(\ell)$ may refer to global and local variables.
The program $\P$ satisfies the specification $\langle\mathit{pre},\mathit{assert}\rangle$ if for all numbers of threads $n$,
the following holds:
For every execution $\langle \vec\ell^{(1)}, s_1\rangle \xrightarrow{\inst{\st_1\;}{\;i_1}} \ldots \xrightarrow{\inst{\st_n\;}{\;i_n}}\langle \vec\ell^{(m)}, s_m\rangle$ of the program instance $\P(n)$,
such that $s_1|_i \models \mathit{pre}$ for all $i\in\{1,\ldots,n\}$
and such that $\mathit{assert}(\ell_j^{(m)})$ is defined for some $j\in\{1,\ldots,n\}$,
we have that $s_m|_j \models \mathit{assert}(\ell_j^{(m)})$.
In the remainder of the paper, we always assume that a parameterized program is accompanied by a specification $\langle \mathit{pre},\mathit{assert}\rangle$.
For examples, we annotate the specification in the thread template (as in \cref{fig:inc-dec-template}).
We simply say that \emph{$\P$ is correct} if $\P$ satisfies this specification.

As an aside, our approach can be extended to more general notions of (safety) specifications,
e.g. a set of error states given by a \emph{generator set} as in~\cite{popl17:thread-modular}.
Such specifications allow for instance a direct encoding of mutual exclusion.
However, since this is orthogonal to our contributions, we focus here on the simpler notion of specification as defined above.

In \cref{sec:red-proofs}, as well as several examples, we consider a particular notion of \emph{proofs} for parameterized programs: \emph{Ashcroft invariants}.
An Ashcroft invariant is a formula of the form
\[
  \forall i_1,\ldots,i_k \,.\, (\bigwedge_{1\leq r< s\leq n} i_r \neq i_s) \to \varphi
\]
where $\varphi$ is a quantifier-free formula,
whose variables range over the global program variables,
indexed local variables $x_{i_r}$ (for $x\in\locals$, $r\in\{1,\ldots,k\}$)
and variables $\pc{i_r}$ (for $r\in\{1,\ldots,k\}$) representing the current control location of thread $i_r$.
The quantified variables symbolically represent $k$ threads of the program.
The premise $\bigwedge_{1\leq r< s\leq n} i_r \neq i_s$ expresses the fact that $i_1,\ldots,i_k$ indeed refer to $k$ distinct threads.
Thus, the conclusion $\varphi$ expresses a relation between the global variables, as well as the locations and local variables of any subset of $k$ distinct threads of the program.
We call the number of quantified variables $k$ the \emph{width} of the Ashcroft invariant.

An Ashcroft invariant is \emph{inductive} for the parameterized program $\P$,
if it is an inductive invariant for every instance $\P(n)$,
assuming the precondition $\mathit{pre}$ initially holds for every thread.
Since we only consider inductive Ashcroft invariants, we omit the adjective from now on.

Finally, let us define what it means for an Ashcroft invariant to prove correctness of a parameterized program $\P$.
We say that an Ashcroft invariant $\forall i_1,\ldots,i_k \,.\, (\bigwedge_{r\neq s} i_r \neq i_s) \to \varphi$ is \emph{safe},
if it is inductive,
and for every location $\ell$ where $\mathit{assert}(\ell)$ is defined,
the following entailment holds:
\[
  \forall i_1,\ldots,i_k \,.\, (\bigwedge_{1\leq r< s\leq n} i_r \neq i_s) \to \varphi \models \forall i\,.\, \pc{i} = \ell \to \mathit{assert}(\ell)
\]
If a safe Ashcroft invariant for a program $\P$ and a specification $\langle \mathit{pre},\mathit{assert}\rangle$ exists,
then $\P$ satisfies the specification $\langle \mathit{pre},\mathit{assert}\rangle$.
However, the reverse is not true.

\paragraph{Other Program Models}
The model of parameterized programs is a natural model for certain classes of concurrent programs,
e.g.\ GPU code and distributed protocols.
More generally, most classes of concurrent programs can be \emph{encoded} in parameterized programs.
Hence our theoretical results can be expected to hold for a wide class of concurrent programs.
In practical terms, such encodings may present a challenge for verification algorithms.
For example, for structured parallel programs with sophisticated dependence graphs implemented using fork/join,
the best practice would not be to encode the program in this model and try to verify it with our verification algorithm.
The main burden in these cases is that the inductive invariant for the program may have to recover part or all of the structure lost from the original model,
and this can be unreasonable to expect from an automated invariant generator.
Smaller extensions of the model, such as allowing a finite number of different thread templates, as in \cref{fig:notify},
are more straightforward and are indeed supported by our implementation.

\section{Reductions Of Parameterized Programs}
\label{sec:param-red}

In this section, we discuss commutativity-based reductions.
We introduce the underlying formalism, which has previously been used for fixed-thread programs, and discuss how it generalizes to parameterized programs.
Then we present our first key contribution:
a finite representation of an infinite family of commutativity-based reductions.
\medskip

To begin, let us quickly summarize the basics of commutativity theory.
The most fundamental notion is a \emph{commutativity relation} between statements.
Specifically, in this work we say that two (indexed) statements $\inst{\st_1}{i}$ and $\inst{\st_2}{j}$ (with $i \neq j$) \emph{commute},
denoted $\inst{\st_1}{i} \comm \inst{\st_2}{j}$,
if executing them in either order yields the same semantics,
i.e., $\sem{\inst{\st_1}{i} \ \ \inst{\st_2}{j}} = \sem{\inst{\st_2}{j} \ \ \inst{\st_1}{i}}$.
We discuss broader notions of sound commutativity in \cref{sec:other-comm}.

The commutativity relation over statements defines an equivalence relation on traces.
We say that two traces $\tau_1$ and $\tau_2$ are \emph{equivalent}
if $\tau_2$ can be derived from $\tau_1$ by repeatedly swapping adjacent commuting statements.
Note that, by repeated application of the definition of commutativity, equivalent traces have the same semantics.
Consequently, it suffices to show that one trace satisfies a specification in order to conclude that all equivalent traces are correct as well.

Motivated by this observation, one can introduce the concept of a \emph{reduction}.
A set of traces $L'$ is a \emph{reduction} of another set of traces $L$
if $L' \subseteq L$, and for each trace in $L$ there exists an equivalent trace in $L'$.
It follows that if we prove that all traces in a reduction $L'$ are correct, we can soundly conclude that all traces in the set $L$ are correct.
Specifically, we are interested in reductions of the language of traces given by a program instance $\P(n)$ for a fixed number of threads $n$.

\subsection{A Family of Reductions}
\label{sec:red-parameterized}

It has been shown that commutativity-based reduction can lead to simpler proofs for concurrent programs with a fixed number of threads.
In particular, the proof for a (suitably chosen) reduction of a program may be within reach of algorithmic verification,
whereas a proof for the entire program may not.

\begin{example}
  \label{ex:inc-dec-fixedthreads}
  Let us consider the program $\P^\pm$ as discussed in the introduction, with the template shown in \cref{fig:inc-dec-template}.
  For any fixed number of threads $n$, the instance $\P^\pm(n)$ is correct.
  In this case, the proof for the (unreduced) program is comparatively simple:
  The instance $\P^\pm(n)$ can be proven correct with the assertions $\texttt{x}\geq0$, $\texttt{x}\geq1$, \ldots, up to $\texttt{x}\geq n$.
  Note however that the proof size, i.e., the required number of assertions, grows with the number of threads.

  Since the increment and decrement of \texttt{x} commute,
  as do two increments resp.\ two decrements,
  we can apply commutativity to simplify the proof.
  We define, for each number of threads $n$, a reduction $\R^\pm(n)$:
  a set of traces that contains, for each equivalence class of traces in $\P^\pm(n)$,
  the representative trace in which each thread executes all its statements in the trace in a single block.
  Thus $\R^\pm(n)$ can be written as $\R^\pm(n) = L_1 L_2 \ldots L_n$, where $L_i = \big( \stsmcol{\inst{\texttt{x:=x+1}\,}{$i$}}{th1} \; \stsmcol{\inst{\texttt{x:=x-1}\,}{$i$}}{th1}\big)^* \big(\varepsilon + \stsmcol{\inst{\texttt{x:=x+1}\,}{$i$}}{th1}\big)$.
  In traces of this reduction, the value of \texttt{x} reaches a value $\geq 2$ only if the last statement executed by some thread $i$ is an increment
  without a matching decrement.
  In this case, \texttt{x} never falls below $2$ again, as every future decrement is preceded by a matching increment.
  Consequently, the resulting reduction $\R^\pm(n)$ can be proven correct with only the assertions $\texttt{x}\geq0$, $\texttt{x}\geq1$, $\texttt{x}\geq 2$, for any number of threads $n$.
\end{example}

In this work, we are concerned with proof simplification for \emph{parameterized} concurrent programs, with an unbounded number of threads.
Thus, we are searching for one uniform proof that proves a program $\P$ correct for all numbers of threads $n$.
A key insight is that commutativity can similarly lead to proof simplification in this setting.

Specifically, suppose that for each $n$, we have proven correctness of a reduction $\R(n)$ of the program instance $\P(n)$ with $n$ threads.
Then, by soundness of commutativity for a fixed number of threads, we can conclude that each $\P(n)$ is correct, i.e., the parameterized program $\P$ is correct.
Furthermore, if the proofs of reductions for different $n$ have a similar structure,
we can hope to find one uniform, finite proof for the parameterized program $\P$.

\begin{example}[continued from \cref{ex:inc-dec-fixedthreads}]
  \label{ex:inc-dec-proofspace}
  Let us consider again the program $\P^{\pm}$,
  and the claim that each reduction $\R^\pm(n)$ can be proven correct with the assertions $\texttt{x}\geq 0$, $\texttt{x} \geq 1$ and $\texttt{x}\geq 2$.
  Specifically, each trace in the reduction $\R^\pm(n)$ can be given a correctness proof (an annotation of the trace)
  using the following Hoare triples, instantiated for all $i\in\{1,\ldots,n\}$:
  \begin{gather*}
    \{\texttt{x}\geq 0\}\,\stsmcol{$\inst{\texttt{x:=x+1}}{i}$}{th1}\,\{\texttt{x}\geq 1\},\qquad
    \{\texttt{x}\geq 1\}\,\stsmcol{$\inst{\texttt{x:=x+1}}{i}$}{th1}\,\{\texttt{x}\geq 2\},\\
    \{\texttt{x}\geq 2\}\,\stsmcol{$\inst{\texttt{x:=x+1}}{i}$}{th1}\,\{\texttt{x}\geq 2\},\qquad
    \{\texttt{x}\geq 1\}\,\stsmcol{$\inst{\texttt{x:=x-1}}{i}$}{th1}\,\{\texttt{x}\geq 0\}\,
  \end{gather*}
  The proof simplification is significant:
  Without reduction, a proof of the program $\P^\pm$ requires a ghost variable
  that counts the number of threads that have incremented but not yet decremented \texttt{x}.
\end{example}

Up to this point, the basis for our considerations has been an infinite family of reductions $\R(n)_{n\in\N}$.
In order to arrive at an effective proof method for parameterized programs,
one crucial step is missing:
We need a way to effectively construct a finite representation of this family.

\subsection{Parameterized Reductions}
The key insight behind our first contribution is this:
\begin{observation}
  For every parameterized program $\P$,
  there exists an infinite family of reductions $\R(n)_{n\in\N}$
  such that the entire family can again be represented as a parameterized program.

\end{observation}

Representing a family of reductions as a parameterized program
enables us to reuse the many mature existing methods for verification of parameterized programs,
and to combine them with commutativity-based reduction.

For a fixed $n$,
a finite automaton recognizing a reduction $\R(n)$ can be constructed using the concept of \emph{sleep sets}~\cite{pldi22:sound-seq}:
In addition to the control locations of the threads, the \emph{sleep set automaton} tracks a set of (indexed) program statements,
the eponymous sleep set.
In each state, the sleep set automaton prevents transitions labeled by statements in the state's sleep set.
After each transition, the sleep set is updated,
i.e., statements are removed and added depending on their commutativity with the statement labeling the transition.
Consider the illustration of an automaton for $\R(2)$ in \cref{fig:sleep-set-illu}.
Initially, the sleep set is empty.
When traversing the edge labeled \stsmcol{$\inst{\st_2}{2}$}{th2},
we add \stsmcol{$\inst{\st_1}{1}$}{th1} to the
\begin{wrapfigure}[12]{r}{0.4\textwidth}
\centering
\tikzstyle{sleep}=[right,font=\footnotesize,color=blue]
\tikzstyle{letter}=[above,font=\footnotesize]
  \begin{tikzpicture}[thick]
    \node[draw,circle,label={[sleep]$\{\}$}] (0) {};
    \node[draw,circle,label={[sleep]$\{\inst{\st_1}{1}\}$},below right=of 0] (1) {};
    \node[draw,circle,label={[sleep]$\{\inst{\st_1}{1}\}$},below right=of 1] (2) {};
    \node[below left=of 0,inner sep=0] (rest) {\vdots};

    \draw[->] (0) -- node[auto,swap]{\stsmcol{$\inst{\st_1}{1}$}{th1}} (rest.north);
    \draw[->] (0) -- node[auto]{\stsmcol{$\inst{\st_2}{2}$}{th2}} (1);
    \draw[-] (1) -- node[auto,swap]{\stsmcol{$\inst{\st_1}{1}$}{th1}} ++(-0.5cm,-0.5cm);
    \draw[-,red] (1)  ++(-0.5cm,-0.5cm) ++(-1mm,1mm) -- ++(2mm,-2mm) ++(-0.5mm,-0.5mm) -- ++(-2mm,2mm);
    \draw[->] (1) -- node[auto]{\stsmcol{$\inst{\st_3}{2}$}{th2}} (2);
    \draw[-] (2) -- node[auto,swap]{\stsmcol{$\inst{\st_1}{1}$}{th1}} ++(-0.5cm,-0.5cm);
    \draw[-,red] (2)  ++(-0.5cm,-0.5cm) ++(-1mm,1mm) -- ++(2mm,-2mm) ++(-0.5mm,-0.5mm) -- ++(-2mm,2mm);
  \end{tikzpicture}
  \caption{Illustration of the sleep set mechanism.
    Sleep sets (blue) are updated after every edge, and lead to removal of transitions.}
\label{fig:sleep-set-illu}
\end{wrapfigure}
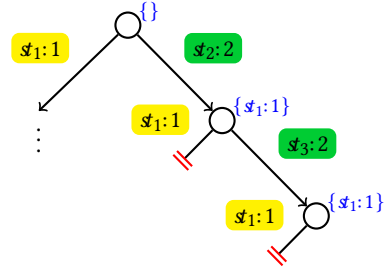
sleep set,
because it has a smaller thread index than \stsmcol{$\inst{\st_2}{2}$}{th2},
and because we assume in this example that \stsmcol{$\inst{\st_1}{1}$}{th1} and \stsmcol{$\inst{\st_2}{2}$}{th2} commute.
In the next state, as \stsmcol{$\inst{\st_1}{1}$}{th1} is in the sleep set,
we prune the corresponding edge.
Any trace that would be accepted via this edge is equivalent to a trace where \stsmcol{$\inst{\st_2}{2}$}{th2} and \stsmcol{$\inst{\st_1}{1}$}{th1} are swapped,
and this trace is already accepted by a run via the left-most \stsmcol{$\inst{\st_1}{1}$}{th1}-transition.
If we assume that \stsmcol{$\inst{\st_1}{1}$}{th1} also commutes with \stsmcol{$\inst{\st_3}{2}$}{th2},
we can keep \stsmcol{$\inst{\st_1}{1}$}{th1} in the sleep set after traversing \stsmcol{$\inst{\st_3}{2}$}{th2},
and again prune the corresponding edge in the right-most state.
If on the other hand \stsmcol{$\inst{\st_1}{1}$}{th1} and \stsmcol{$\inst{\st_3}{2}$}{th2} did not commute,
we would instead remove \stsmcol{$\inst{\st_1}{1}$}{th1} from the sleep set and preserve the transition.

The sleep set technique as explained here can be applied for any fixed number of threads~$n$.
However, each $n$ yields a different language, and this approach does not lead to a uniform representation for the family of reductions.
The key insight which enables such a uniform finite representation is the observation that for the correctness,
we are only interested in \emph{feasible} traces.
Thus, we can encode the family of reductions through an \emph{instrumentation} of the original program's thread template.
For each thread instance $i$, we add a boolean variable $\sleep{i}$,
which keeps track of whether thread $i$ (resp.\ its currently enabled statements) are in the sleep set.
Consequently, when $\sleep{i}$ is true, thread $i$ must not make a move.
In other words, any trace where thread $i$ makes a move while $\sleep{i}$ is true must be \emph{infeasible}.
By shifting from an \emph{explicit} mechanism (computing sleep sets, and removing edges from an automaton) to a \emph{symbolic} approach,
we thus arrive at a uniform finite representation of the family of reductions.

This instrumentation deviates slightly from the explanation above.
Instead of tracking statements in the sleep set,
we track the threads that would execute these statements.
I.e., the variable $\sleep{i}$ of a thread $i$ is true, if the thread's next enabled statements $\enabled{\ell_i}$ are in the sleep set.
(In case multiple statements are enabled, i.e., at a branch or loop head with enabled statements \stsmcol{\assume{$e$}}{th1} and \stsmcol{\assume{$\lnot e$}}{th1},
either both statements are in the sleep set, or neither is.)
This shift from tracking individual letters in the sleep set to tracking the threads pays off in terms of the complexity added to the state space of the instrumented program:
We need only a single boolean variable, rather than one variable for each statement that appears in the thread template.

We define a formula which expresses that thread $j$ is in a control location whose enabled statements commute with a given statement $\inst{\st}{i}$ executed by a different thread $i$.

\begin{definition}[Commutativity Test]
\label{def:comm-cond}
The \emph{commutativity test} $\commCond{j}{i}{\st}$ is the formula
\[
\commCond{j}{i}{\st} :\equiv \bigvee \{\, \pc{j}=\ell \mid \ell\in\Loc \land \forall\st'\in\enabled{\ell}\,.\, \inst{\st'}{j} \comm \inst{\st}{i} \,\}.
\]
\end{definition}

The commutativity test is used in the instrumentation of statements.

\begin{definition}[Instrumented Statements]
  \label{def:instr-stmt}
  Let $\st$ be a statement.
  We define the instrumented statement $\instrument{\st}$ as the atomically executed block of statements
  \[
    \instrument{\st} := \left[\begin{array}{l}
      \assume{$\lnot\sleep{}$}\\
      \syncStmt{$\sleep{j}$}{$(\sleep{j}\lor \id{j}<\id{i})\land\commCond{j}{i}{\st}$}\\
      \st
    \end{array}\right]
  \]
\end{definition}

An instrumented statement $\instrument{\st}$ first checks if its thread is in the sleep set, and if so, blocks.
Otherwise, i.e., if the statement is allowed to execute, the instrumentation performs the update of the sleep set
through a \emph{synchronized} statement~\cite{popl17:thread-modular} that modifies the $\sleep{}$ variables of all (unboundedly many) other threads $j$.
Finally, the original statement $\st$ executes.
Recall that here, the symbols $i$ and $j$ are part of the syntax of synchronized statements rather than logical variables:
The symbol $i$ represents the thread executing the statement, and $j$ represents any other thread.

Note that the instrumentation refers to a thread-local integer variable $\id{}$.
We add such a (nondeterministically initialized) ID variable to the thread template to serve as a tie-break.
If two threads can move, but allowing both to execute statements would result in equivalent traces, we must identify which thread should go first.
We assume that all thread IDs are pairwise distinct.
Thus, these thread IDs allow us to distinguish the thread instances and to decide:
If the enabled statements of two threads commute, the thread with a smaller ID moves first.

\begin{definition}[Sleep-Instrumented Program]
  Let $\P= \langle \Loc, \Delta, \ell_\init, \locals\rangle$ be a parameterized program. %
  We define the \emph{sleep-instrumented} program $\P_\sleep{} := \langle \Loc, \Delta_\sleep{},\ell_\init, \locals'\rangle$ %
  with local variables $\locals' := \locals \cup \{ \id{}, \sleep{} \}$,
  and the transitions given by
  \[
      \Delta_\sleep{} := \left\{\, \left\langle \ell,\instrument{\st}, \ell' \right\rangle\mid \langle \ell,\st,\ell'\rangle \in \Delta \,\right\}.
  \]
\end{definition}
\begin{example}[Continued from \cref{ex:inc-dec-proofspace}]
\label{ex:inc-dec-instr}
\Cref{fig:inc-dec-template-instr} shows the thread template for the sleep-instrumented program $\sleepInstr{\P^\pm}$ corresponding to the program $\P^\pm$ shown in \cref{fig:inc-dec-template}.
Since all statements of $\P^\pm$ commute, the commutativity tests $\commCondSt{j}{i}{\texttt{x:=x+1}}$ and $\commCondSt{j}{i}{\texttt{x:=x-1}}$
both resolve to the formula $\pc{j} = \ell_0\lor\pc{j}=\ell_1$.
\end{example}

\begin{figure}
 \vspace{-3.5mm}
  \begin{tikzpicture}[thick,font=\scriptsize,node distance=10cm]
    \node[draw,circle,label={[]above:$\ell_0$}] (0) {};
    \node[draw,circle,label={[]above:$\ell_1$},label={below:\assert{x!=0}},right of=0] (1) {};
    \node[left=5mm of 0,inner sep=0] (pre) {$\{\texttt{x}=0\}$};
    \draw[<-] (0) -- (pre);
    \draw[->] (0) edge[bend left=12] node[auto]{$\left[\begin{array}{l}
        \assume{$\lnot\sleep{}$}\\
        \syncStmt{$\sleep{j}$}{$(\sleep{j}\lor \id{j}<\id{i})\land(\pc{j}=\ell_0\lor\pc{j}=\ell_1)$}\\
        \texttt{x:=x+1}
      \end{array}\right]$} (1);
    \draw[->] (1) edge[bend left=12] node[auto]{$\left[\begin{array}{l}
        \assume{$\lnot\sleep{}$}\\
        \syncStmt{$\sleep{j}$}{$(\sleep{j}\lor \id{j}<\id{i})\land(\pc{j}=\ell_0\lor\pc{j}=\ell_1)$}\\
        \texttt{x:=x-1}
      \end{array}\right]$} (0);
  \end{tikzpicture}
  \caption{Template for $\sleepInstr{\P^{\pm}}$}
  \label{fig:inc-dec-template-instr}
\end{figure}

The sleep-instrumented program $\sleepInstr{\P}$ serves as uniform finite representation for the family of reductions $\R(n)_{n\in\N}$.
To formalize this relationship,
let $\pi$ be the inverse of $\iota$, i.e., a mapping between statements such that $\pi(\iota(\st)) = \st$.
We extend this mapping to indexed statements, traces, and sets of traces in the natural way.
The following key result formally expresses that the sleep-instrumented program describes a family of reductions $\R(n)_{n\in\N}$ of the original program, modulo feasibility.
\begin{theorem}[Reduction]
  \label{thm:instr-feas-red}
  Let $\mathbf{Feas}$ be the set of all feasible traces.
  For each number of threads $n$, the set of traces $\pi(\sleepInstr{\P}(n) \cap \mathbf{Feas})$ is a reduction of $\P(n)\cap\mathbf{Feas}$,
  i.e., the feasible traces of $\P(n)$.%
\end{theorem}
We consider only feasible traces, since the reduction works based on the guards ($\lnot\sleep{}$) added to each transition.
This is necessary to describe the reduction as a parameterized program:
Only when we fix the number of threads $n$, the sleep guards and updates can be evaluated.

\begin{theorem}[Soundness]
  \label{thm:instr-sound-red}
  The sleep-instrumented program $\sleepInstr{\P}$ is correct iff $\P$ is correct.
\end{theorem}

The reduction achieved by the instrumentation is \emph{minimal}:
We retain only one representative per equivalence class,
and hence a strict subset cannot be a reduction.
This means that we do not unnecessarily burden the verification with the proof of redundant traces;
the instrumentation fully realizes the benefit of commutativity.

\begin{proposition}[Minimality]
  \label{prop:minimal}
  For every feasible trace $\tau$ of $\P$, the sleep-instrumented program $\sleepInstr{\P}$ has exactly one feasible trace $\tau'$ such that $\pi(\tau')$ is equivalent to $\tau$.
\end{proposition}

As demonstrated in \cref{sec:motivating-example},
there exist programs such that no proof of the program without non-trivial ghost state exists,
but where some reduction of the program has a simple proof.
We investigate this phenomenon for the sleep-instrumented program.
To make this precise, we fix Ashcroft invariants as our notion of \emph{proof},
and consider a simple example.

{\tiny~}
\begin{wrapfigure}[9]{r}{0.53\textwidth}
  \footnotesize
  \vspace{-2.75em}
  \begin{minipage}{0.53\textwidth}
  \begin{align*}
    \forall i,j\,.\, i \neq j \to \big( &x\geq 0
    \\\null\land &(\pc{i} = \ell_1 \to x \geq 1)
    \\\null\land &(\pc{j} = \ell_1 \to x \geq 1)
    \\\null\land &(\pc{i} = \ell_1 \land \pc{j} = \ell_1 \to x \geq 2 \lor (\sleep{i} \land \sleep{j}))
    \\\null\land &(\id{i} < \id{j} \land \pc{j} = \ell_1 \to \sleep{i})
    \\\null\land &(\id{i} > \id{j} \land \pc{i} = \ell_1 \to \sleep{j})\big)
  \end{align*}
  \end{minipage}
    \vspace{-0.75em}
  \caption{Safe Ashcroft invariant for $\P^{\pm}_\sleep{}$}
  \label{fig:inc-dec-sleep-ashcroft}
\end{wrapfigure}
\vspace*{-2em}
\begin{example}[continued from \cref{ex:inc-dec-instr}]
\label{ex:inc-dec-red-ashcroft}
  Consider again the program $\P^{\pm}$, with the template shown in \cref{fig:inc-dec-template}.
  There does not exist a safe Ashcroft invariant of any width for the program $\P^{\pm}$~\cite{popl17:thread-modular}.
  Intuitively, the invariant would have to express the information that the value of the global variable $x$ is always greater than or equal to the number of threads that have executed the increment but not yet the decrement.
  Ashcroft invariants cannot express this information (for a formal argument, see \cite{popl17:thread-modular}).

  \Cref{fig:inc-dec-sleep-ashcroft} shows a safe Ashcroft invariant (of width 2) for the sleep-instrumented program $\sleepInstr{\P^{\pm}}$ (shown in \cref{fig:inc-dec-template-instr}).
  This invariant uses the fact that, in the traces of the reduction, if the value of \texttt{x} exceeds 2, it never falls below $2$ again.
  Any trace of the original program is equivalent to a trace in the reduction, since increments and decrements commute and can be arbitrarily reordered.

  Let us examine some traces of $\sleepInstr{\P^\pm}(2)$
  to see how the Ashcroft invariant proves the correctness of traces in $\R^\pm(2)$ as well as outside $\R^\pm(2)$ by using the variables $\sleep{i}$ and $\sleep{j}$.
  We assume that $\id{1}<\id{2}$.
  Given a trace of $\sleepInstr{\P^\pm}(2)$, we instantiate $i:=1$ and $j:=2$ in the Ashcroft invariant, insert concrete values for $\pc{1}$ and $\pc{2}$,
  and simplify the formula to get an inductive annotation of the trace.
  \newcommand\annot[1]{{\color{red}\{#1\}}}
  For instance, the trace \stsmcol{$\inst{\texttt{x:=x+1}}{1}$}{th1}\,\stsmcol{$\inst{\texttt{x:=x+1}}{2}$}{th2}\,\stsmcol{$\inst{\texttt{x:=x-1}}{1}$}{th1} of $\P^\pm(2)$
    is not included in the reduction $\R^\pm(2)$.
  We get the following annotation for the corresponding instrumented trace:
  \[
    \annot{\texttt{x}\geq 0}\ \stsmcol{$\inst{\iota(\texttt{x:=x+1})}{1}$}{th1}
    \ \annot{\texttt{x}\geq 1}\ \stsmcol{$\inst{\iota(\texttt{x:=x+1})}{2}$}{th2}
    \ \annot{\texttt{x}\geq 1 \land (\texttt{x} \geq 2 \lor \sleep{2}) \land \sleep{1}}\ \stsmcol{$\inst{\iota(\texttt{x:=x-1})}{1}$}{th1}
    \ \annot{\texttt{x}\geq 1}
  \]
  Consider in particular the last Hoare triple.
  Since $\sleep{1}$ holds, and \stsmcol{$\inst{\iota(\texttt{x:=x-1})}{1}$}{th1} assumes $\lnot\sleep{1}$,
  the last statement cannot be executed, i.e., the trace is infeasible (hence, any postcondition holds afterwards).
  By contrast, consider the annotated trace corresponding to \stsmcol{$\inst{\texttt{x:=x+1}}{1}$}{th1}\,\stsmcol{$\inst{\texttt{x:=x+1}}{2}$}{th2}\,\stsmcol{$\inst{\texttt{x:=x-1}}{2}$}{th2}:
  \[
    \annot{\texttt{x}\geq 0}\ \stsmcol{$\inst{\iota(\texttt{x:=x+1})}{1}$}{th1}
    \ \annot{\texttt{x}\geq 1}\ \stsmcol{$\inst{\iota(\texttt{x:=x+1})}{2}$}{th2}
    \ \annot{\texttt{x}\geq 1 \land (\texttt{x} \geq 2 \lor \sleep{2}) \land \sleep{1}}\ \stsmcol{$\inst{\iota(\texttt{x:=x-1})}{2}$}{th2}
    \ \annot{\texttt{x}\geq 1}
  \]
  Note again the last Hoare triple.
  The assumption $\lnot\sleep{2}$ by the statement \stsmcol{$\inst{\iota(\texttt{x:=x-1})}{2}$}{th2}
  together with the precondition ensures that $\texttt{x}\geq 2$, and thus $\texttt{x}\geq 1$ still holds after the decrement.
  The final postcondition $\texttt{x}\geq 1$ however does not prevent us from extending the trace with the statement \stsmcol{$\inst{\iota(\texttt{x:=x-1})}{1}$}{th1},
  yielding the postcondition $\texttt{x}\geq 0$.
  While the resulting trace would not correspond to a trace in the reduction $\R^\pm(2)$,
  the Ashcroft invariant does not prove its infeasibility.
  Instead, it simply proves that the trace satisfies the specification.
\end{example}

Even in cases where a safe Ashcroft invariant of some width $k$ exists,
sleep instrumentation can simplify the proof.
Specifically, sleep instrumentation can reduce the minimum width of a safe Ashcroft invariant.
\begin{example}
\label{ex:decrease-width}
  Consider the program $\P^{K\pm}$ with the template shown in \cref{fig:hierarchy-collapse}, where \texttt{x} is a global integer variable.
  Given a fixed value for the constant $K$, a safe Ashcroft invariant for this program must have at least width $K+1$~\cite{popl17:thread-modular}.
  However, the Ashcroft invariant shown in \cref{fig:inc-dec-sleep-ashcroft} has width~$2$, and is safe for the sleep-instrumented program $\sleepInstr{\P^{K\pm}}$, for every value of $K$.
\end{example}

The following theorem states that, if we already have a proof (i.e., a safe Ashcroft invariant) for the original program $\P$,
there also exists a safe Ashcroft invariant for the sleep-instrumented program.
Hence, more (and by \cref{ex:inc-dec-red-ashcroft}, \emph{strictly} more) programs can be proven correct with our instrumentation than without.
Additionally, the theorem shows that a proof of the sleep-instrumented program $\sleepInstr{\P}$ need never be more complicated than a proof of the original program~$\P$;
and in cases such as \cref{ex:decrease-width}, it may be strictly simpler.
\begin{theorem}[Conservative Extension]
  \label{thm:conservative-instr}
  Every safe Ashcroft invariant for a program $\P$
  is a safe Ashcroft invariant for the corresponding sleep-instrumented program $\sleepInstr{\P}$.
\end{theorem}

\section{Reductions Beyond Sequential Composition}
\label{sec:pref-orders}

Up to this point,
we have considered a very restricted class of reductions based on thread ordering:
A thread $i$ could only change $\sleep{j}$ from false to true if $\id{j} < \id{i}$.
If all statements of different threads commute, the resulting reduction is the sequential composition of threads:
As soon as a thread with a higher ID takes a step, all threads with lower ID are ``put to sleep'' and never awakened again.
If we do not have total commutativity, the reduction overapproximates sequential composition.

Recall the program $\P^\mathrm{notify}$ discussed in \cref{sec:motivating-example} (as shown \cref{fig:notify}).
For this program it is crucial to align the ``sends'' (i.e., writes to the \texttt{queue} array) in the \texttt{notifier} thread
with all the ``receives'' (i.e., reads from the \texttt{queue} array) in the \texttt{listener} threads in order to find a simple proof.
The (approximation of) sequential composition would not provide sufficient opportunity for simplification.
Therefore, in this section, we widen our view to consider the larger class of \emph{lexicographical reductions}~\cite{pldi22:sound-seq},
which have been shown to be practically useful reductions for program verification.

Previous work~\cite{pldi22:sound-seq} uses \emph{preference orders} to describe different reductions of fixed-thread programs.
A preference order is a total order over program traces (or, more generally, words over some alphabet).
It can be used to define a reduction as follows:
\begin{definition}[Definition~4.2 in \citep{pldi22:sound-seq}]
Let $L$ be a language over an alphabet $\Sigma$,
and let $\preceq$ be a total order over $\Sigma^*$.
The \emph{reduction of $L$ induced by $\preceq$} is denoted $\mathit{red}_{\preceq}(L)$ and contains, for each equivalence class, only the minimal trace wrt.\ the preference order.
\[
  \mathit{red}_{\preceq}(L) =\{\,{\min}_\preceq[w] \mid w \in L \,\}
\]
\end{definition}
In this work, we focus on the class of \emph{positional lexicographic preference orders}~\cite{pldi22:sound-seq}.
Positional lexicographic preference orders are a generalization of a lexicographic orders over program traces,
where the underlying order on statements may differ depending on the current program locations of all threads.
The reductions induced by positional lexicographic preference orders are called \emph{lexicographical reductions}.

We extend the concept of (positional lexicographic) preference orders to parameterized programs.
\begin{definition}[Parameterized Preference Order]
  A \emph{parameterized preference order} is a family of functions $(\preccurlyeq^n)_{n\in\N}$,
  where $\preccurlyeq^n : \mathbf{Loc}^n \to \mathit{TO}_{\{1,\ldots,n\}}$ maps $n$-tuples of locations $\vec{\ell}$ to total orders $\preccurlyeq^n_{\vec{\ell}}$ over thread indices $\{1,\ldots,n\}$.
  For $\vec{\ell}\in\Loc^n$, we write $\preccurlyeq_{\vec\ell}$ instead of $\preccurlyeq^n_{\vec\ell}$ (we omit the superscript $n$).
\end{definition}

A parameterized preference order is thus given by the choice of the underlying ordering of threads
(all statements of the same thread are ordered the same).
As the threads move to different control locations, the ordering of threads assigned by a parameterized preference order may change.
Thus the reduction may differ significantly from the sequential composition of threads.

We focus on a subclass of finitely describable parameterized preference orders:
\begin{definition}[Pairwise Preference Order]
  \label{def:spo}
  Let $R \subseteq \Loc^2$ be a total %
  and transitive relation.
  The \emph{pairwise preference order} induced by $R$
  is the parameterized preference order $(\preccurlyeq^n)_{n\in\N}$ such that
  \[
    i \preccurlyeq_{\vec\ell} j \quad \iff \quad \langle \ell_i,\ell_j \rangle \in R \land ( \langle \ell_j,\ell_i \rangle\in R \to \id{i} \leq \id{j})
  \]
\end{definition}
Thus, the ordering of threads $i,j$ wrt.\ a pairwise preference order only depends on the control locations $\ell_i,\ell_j$ of $i$ and $j$ in the tuple $\vec{\ell}$.
The locations of other threads do not play a role.
If the pair $\langle \ell_i,\ell_j\rangle$ is in $R$, and the reversed pair $\langle\ell_j,\ell_i\rangle$ is \emph{not} in $R$,
  we prefer thread $i$.
If both the pairs $\langle \ell_i,\ell_j\rangle$ and $\langle\ell_j,\ell_i\rangle$ are in $R$, we prefer the thread with a smaller ID.
By totality of $R$ (i.e., $\langle \ell, \ell'\rangle \in R \lor \langle \ell',\ell\rangle\in R$ for all $\ell,\ell'\in\Loc$), one of two compared threads must always be preferable over the other.
\begin{proposition}
  Each total, transitive relation $R \subseteq \Loc^2$ defines a parameterized preference order.
\end{proposition}
\begin{proof}
  We have to show that for every $n$ and $\vec{\ell}\in\Loc^n$, the induced relation ${\preccurlyeq_{\vec{\ell}}}$ is a total order over the set of thread indices $\{1,\ldots,n\}$.
  \begin{description}
    \item[Reflexivity] Follows from totality, which is shown below.
    \item[Antisymmetry] Let $i \preceq_{\vec{\ell}} j$ and $j \preceq_{\vec{\ell}} i$.
      It follows that $\langle \ell_i,\ell_j \rangle\in R$, $\langle \ell_j,\ell_i \rangle\in R$,
      and, by the respective implications, also $\id{i} \leq \id{j}$ and $\id{j} \leq \id{i}$.
      Thus we have $\id{i}=\id{j}$, and by uniqueness of thread IDs, we conclude $i=j$.
    \item[Transitivity] Let $i \preceq_{\vec{\ell}} j \preceq_{\vec{\ell}} k$.
      Thus we have $\langle \ell_i,\ell_j \rangle\in R$ and $\langle \ell_j,\ell_k \rangle\in R$,
      and we know that the implications $\langle \ell_j,\ell_i \rangle\in R \to \id{i} \leq \id{j}$ and $\langle \ell_k,\ell_j \rangle\in R \to \id{j} \leq \id{k}$ hold.

      By transitivity of $R$ we know that $\langle \ell_i,\ell_k \rangle\in R$.
      It remains to show that the implication $\langle \ell_k,\ell_i \rangle\in R \to \id{i} \leq \id{k}$ holds.
      Suppose that $\langle \ell_k,\ell_i \rangle\in R$.
      By transitivity of $R$, we have $\langle \ell_k,\ell_j \rangle\in R$ and thus $\id{j} \leq \id{k}$.
      Furthermore, again by transitivity, we have $\langle \ell_j,\ell_i \rangle\in R$ and thus $ \id{i} \leq \id{j}$.
      It follows that indeed $\id{i}\leq\id{j}\leq\id{k}$.
    \item[Totality] Let $i,j\in\{1,\ldots,n\}$, and wlog.\ $\id{i}\leq\id{j}$.
      By totality of $R$, we must have $\langle \ell_i, \ell_j \rangle \in R$ or $\langle \ell_j,\ell_i \rangle\in R$.
      If $\langle \ell_i, \ell_j \rangle \in R$, we have that $i \preceq_{\vec{\ell}} j$ (the implication $\langle\ell_j,\ell_i\rangle\in R\to \id{i}\leq\id{j}$ holds, because the conclusion holds).
      Otherwise, if $\langle \ell_j,\ell_i \rangle\in R$ but $\langle \ell_i, \ell_j \rangle \notin R$,
      we have $j \preceq_{\vec{\ell}} i$ (the implication $\langle\ell_i,\ell_j\rangle\in R\to \id{j}\leq\id{i}$ holds, because the premise does not hold).
  \end{description}
\end{proof}

\Cref{sec:param-red} considers the special case that $R=\Loc^2$,
such that we have $i \preccurlyeq_{\vec{\ell}} j \iff \id{i} \leq \id{j}$.
The class of pairwise preference orders also includes other interesting orders.
\begin{example}[Lockstep Order]
  For each $\ell\in\Loc$, let $d(\ell)$ be the minimum length of a path (in the thread template) from the initial location to $\ell$.
  We define the transitive and total relation $R = \{\,\langle\ell,\ell'\rangle\mid d(\ell) \leq d(\ell')\,\}$.
  The induced pairwise preference order mimics lock-step execution:
  Whenever a thread $i$ has ``fallen behind'' a thread $j$ (i.e., $d(\ell_i) < d(\ell_j)$),
  thread $i$ is preferred over thread $j$ and is thus allowed to ``catch up''.
  When the locations of both threads have the same depth, i.e., $\langle \ell_i,\ell_j \rangle \in R$ and $\langle \ell_j,\ell_i \rangle\in R$,
  the thread with the smaller ID is preferred and takes the next step.
\end{example}
Let us once again consider the program $\P^\mathrm{notify}$.
For this program, the reduction which admits the simple proof discussed in \cref{sec:motivating-example} is induced by lockstep order.

The construction of our instrumented program $\sleepInstr{\P}$ can be generalized to arbitrary pairwise preference orders.
To this end, and for the remainder of the paper, let $R\subseteq \Loc^2$ be a total and transitive relation.
In order to represent the reduction wrt.\ any pairwise preference order again as a parameterized program,
we define:
\begin{definition}[Preference Test]
  \label{def:pref-test}
  The \emph{preference test} for the pairwise preference order induced by the total and transitive relation $R \subseteq \Loc^2$ is defined as the following formula:
  \[
    \prefTest{i}{j} :\equiv \langle pc_i, pc_j \rangle \in R \land (\langle pc_j,pc_i\rangle \in R \to \id{i}\leq\id{j})
  \]
\end{definition}
The preference test $\prefTest{i}{j}$ evaluates to true if, in the current program configuration $\langle \vec\ell,s\rangle$,
our pairwise preference order prefers the statements of thread $i$ over the statements of thread $j$, i.e., if $i \preceq_{\vec{\ell}} j$.
We modify the instrumentation of statements (\cref{def:instr-stmt}) to use the preference test.
\begin{definition}[Instrumented Statement with Preference Test]
  Let $\st$ be a statement.
  We define the instrumented statement $\instrument{\st}$ as the atomically executed block of statements
  \[
    \instrument{\st} := \left[\begin{array}{l}
      \assume{$\lnot\sleep{}$}\\
      \syncStmt{$\sleep{j}$}{$(\sleep{j}\lor {\color{red}\prefTest{j}{i}})\land\commCond{j}{i}{\st}$}\\
      \st
    \end{array}\right]
  \]
\end{definition}
Our results in \cref{sec:param-red} (\cref{thm:instr-feas-red,thm:instr-sound-red,prop:minimal,thm:conservative-instr}) still hold for the modified instrumentation, and for every pairwise preference order.

\section{Finding Ashcroft Invariants for a Reduction}
\label{sec:red-proofs}

We apply the approach of \emph{thread-modular verification at many levels}~\cite{popl17:thread-modular} to find proofs of parameterized programs,
in the form of Ashcroft invariants.
We show how this approach can be applied to the sleep-instrumented program $\sleepInstr{\P}$ to verify a reduction of a parameterized program.

\subsection{Thread-Modular Verification of Reductions}
In thread-modular verification at many levels~\cite{popl17:thread-modular},
the existence of a safe Ashcroft invariant of some fixed width $k$ for a program $\P$
is encoded through a constraint Horn clause (CHC) system.
This CHC system, which we denote $\chcTm{\P}{k}$, uses a single uninterpreted predicate symbol $\Inv(g,\pc{1},x_1,\ldots,\pc{k},x_k)$.
The parameter $g$ represents the global variables of the program.
The parameters $\pc{i}$ and $x_i$ represent the current control locations resp.\ the thread-local variables of $k$ different thread instances.

We can apply an off-the-shelf CHC solver to check satisfiability of this CHC system.
If the system is unsatisfiable, there does not exist a safe Ashcroft invariant of width $k$.
However, this does not mean that the program is incorrect.
It might simply be that every safe Ashcroft invariant has a width larger than $k$,
or that there does not exist a safe Ashcroft invariant of any width, yet the program is still correct.
If on the other hand the CHC system is satisfiable, we can construct an Ashcroft invariant from a solution.
\begin{lemma}[Lemmas 1 and 3 in \cite{popl17:thread-modular}]
  \label{lem:chc-ashcroft}
  If $\Phi_\Inv$ is a solution of $\chcTm{\P}{k}$,
  then the formula
  \[
    \forall i_1,\ldots,i_k\,.\, ({\textstyle \bigwedge_{1\leq r< s\leq n} i_r \neq i_s}) \to \Phi_\Inv(g,\pc{i_1},x_{i_1},\ldots,\pc{i_k}, x_{i_k})
  \]
  is a safe Ashcroft invariant (of width $k$) for the program $\P$.
\end{lemma}

\begin{figure}
  \footnotesize
  \begin{flalign}
    \intertext{\parbox{\linewidth}{{\em\bfseries Initial}:}}
    &\begin{aligned}[t]
      &\quad\left({\textstyle\bigwedge_{i=1}^k\mathit{pre}(g,x_1)}\right)
      \land \left( {\textstyle\bigwedge_{i=1}^k \pc{i} = \ell_\init} \right)
      \land {\textstyle\color{red}\left( \bigwedge_{i=1}^k \lnot\sleep{i} \right)}
      \land {\textstyle\color{red}\left(\bigwedge_{i\neq j} \id{i} \neq \id{j}\right)}\\
      &\quad\quad\quad\to \Inv(g,\id{1},\pc{1},\sleep{1},x_1,\ldots,\id{k},\pc{k},\sleep{k},x_k)
    \end{aligned}
    \\
    \intertext{\parbox{\linewidth}{{\em\bfseries Inductivity} (for each edge $\langle\ell,\st,\ell'\rangle \in \Delta$ and each $i\in\{1,\ldots,k\}$):}}
    &\begin{aligned}[t]
      &\quad\Inv(g,\id{1},\pc{1},\sleep{1},x_1,\ldots,\id{k},\pc{k},\sleep{k},x_k)
        \land \pc{i} = \ell \land \pc{i}' = \ell'
        \land \st(g,x_i,g',x_i')\\
      &\quad\quad \land {\color{red}\lnot\sleep{i}}
        \land {\textstyle\color{red} \bigwedge_{j\neq i} \left( \sleep{j}' \leftrightarrow \big(\sleep{j} \lor \prefTest{j}{i} \big) \land \commCond{j}{i}{\st} \right)} \\
      &\quad\quad\quad \to \Inv(g',\id{1},\pc{1},\sleep{1}',x_1,\ldots,\id{i},\pc{i}',\sleep{i},x_i',\ldots,\id{k},\pc{k},\sleep{k},x_k)
    \end{aligned}
    \label{eq:symbolic-ind}
    \\
    \intertext{\parbox{\linewidth}{{\em\bfseries Non-Interference} (for each edge $\langle\ell,\st,\ell'\rangle\in\Delta$):}}
    &\begin{aligned}[t]
      &\quad\Inv(g,\id{1},\pc{1},\sleep{1},x_1,\ldots,\id{k},\pc{k},\sleep{k},x_k)\\
      &\quad\quad \land \Inv(g,\id{\star},\pc{\star},\sleep{\star},x_\star, \ldots, \id{k},\pc{k},\sleep{k},x_k)
        \land \ldots
        \land \Inv(g,\id{1},\pc{1},\sleep{1},x_1,\ldots,\id{\star},\pc{\star},\sleep{\star},x_\star)\\
      &\quad\quad \land \pc{\star} = \ell \land \pc{\star}' = \ell' \land \st(g,x_\star,g',x_\star')\\
      &\quad\quad \land {\color{red}\lnot\sleep{\star}}
        \land {\textstyle\color{red} \bigwedge_{j} \left( \sleep{j}' \leftrightarrow \big(\sleep{j} \lor \prefTest{j}{\star} \big) \land \commCond{j}{\star}{\st} \right)} \\
      &\quad\quad\quad \to \Inv(g',\id{1},\pc{1},\sleep{1}',x_1,\ldots,\id{k},\pc{k},\sleep{k}',x_k)
    \end{aligned}
    \label{eq:symbolic-nonint}
    \\
    \intertext{\parbox{\linewidth}{{\em\bfseries Safety} (for each $i\in\{1,\ldots,k\}$ and $\ell\in\Loc$ where $\mathit{assert}(\ell)$ is defined):}}
    &\quad\Inv(g,\id{1},\pc{1},\sleep{1},x_1,\ldots,\id{k},\pc{k},\sleep{k},x_k) \land \pc{i} = \ell \land \lnot \mathit{assert}(\ell) \to \bot
  \end{flalign}
  \caption{%
    \emph{Symbolic-sleep} CHC encoding $\chcSymb{\P}{k}$ for the existence of a safe Ashcroft invariant of width $k$ for the sleep-instrumented program $\sleepInstr{\P}$.
    Differences from the CHC encoding $\chcTm{\P}{k}$ for the unreduced program are highlighted in red.
  }
  \label{fig:chc-symbolic}
\end{figure}

We apply the same methodology to the sleep-instrumented program $\P_\sleep{}$,
yielding the CHC system $\chcSymb{\P}{k}$.
Since $\P_\sleep{}$ is again a parameterized program,
no conceptual changes are required.
In particular, the thread-modular CHC encoding supports the \emph{synchronized statements} used by our instrumentation
to update the sleep variables of all (unboundedly many) other threads~\cite{popl17:thread-modular}.
\Cref{fig:chc-symbolic} shows the resulting CHC encoding for $\chcSymb{\P}{k}$.
We call the encoding $\chcSymb{\P}{k}$ the \emph{symbolic-sleep} encoding, to distinguish it from the \emph{explicit-sleep} encoding introduced in \cref{sec:break-symm}.

Intuitively, the clauses describe an invariant predicate $\Inv$ that must hold for any subset of $k$ distinct threads
(mirroring the structure of Ashcroft invariants).
The clause {\em\bfseries Initial} establishes that the invariant holds initially.
For any $k$ threads $i_1, \ldots, i_k$,
the {\em\bfseries Inductivity} clauses demand that the invariant must be preserved if any of the threads $i_1,\ldots,i_k$ makes a step,
whereas {\em\bfseries Non-Interference} imposes that the invariant is preserved if another thread (denoted $\star$) makes a step.
Replacing any of the threads $i_r$ by the thread $\star$ yields another set of $k$ distinct threads,
so we may assume that the invariant $\Inv$ holds for any of these sets.
This yields the additional premises in the {\em\bfseries Non-Interference} clause.
Finally, {\em\bfseries Safety} ensures that a solution to the CHC system describes a \emph{safe} Ashcroft invariant.

\begin{proposition}
  \label{prop:symbolic-sound}
  If the CHC system $\chcSymb{\P}{k}$ is satisfiable,
  the program $\P$ is correct.
\end{proposition}
\begin{inlineproof}
  Follows from \cref{lem:chc-ashcroft} and \cref{thm:instr-sound-red}.
\end{inlineproof}

In analogy to \cref{thm:conservative-instr},
the symbolic-sleep CHC encoding $\chcSymb{\P}{k}$ behaves conservatively wrt.\ the encoding $\chcTm{\P}{k}$ of the original program.
\begin{observation}
\label{obs:conservative-chc-symb}
Any solution to $\chcTm{\P}{k}$ is also a solution to $\chcSymb{\P}{k}$.
Moreover, there are cases where $\chcTm{\P}{k}$ has no solution, but $\chcSymb{\P}{k}$ does.
\end{observation}

\subsection{Breaking Symmetry with the Explicit-Sleep Encoding}
\label{sec:break-symm}

Despite \cref{obs:conservative-chc-symb},
it is not clear that a CHC solver will be faster to find a solution when applied to the symbolic-sleep encoding $\chcSymb{\P}{k}$
than when applied to the encoding $\chcTm{\P}{k}$.
In order to gain a systematic understanding of how easy or difficult it is for a CHC solver to find a solution,
let us introduce the notions of \emph{search} and \emph{solution space}.

\begin{definition}[Search and Solution Space]
  Let $\mathcal{C}$ be a CHC system over a single predicate symbol $p(v_1,\ldots,v_m)$ of arity $m$.
  The search space $\search{\mathcal{C}}$ of possible solutions to $\mathcal{C}$
  is the set of all first-order formulae $\Phi_p(v_1,\ldots,v_m)$ whose free variables lie in $\{\,v_1,\ldots,v_m\,\}$.

  The solution space $\sol{\mathcal{C}}$ denotes the subset of the search space $\search{\mathcal{C}}$ containing exactly all those predicates $\Phi_p(v_1,\ldots,v_m)$ that satisfy the given CHC system $\mathcal{C}$.
\end{definition}

A larger solution space means that a solver is more likely to find a satisfying solution to a CHC system,
whereas a larger search space is indicative of potential additional effort to rule out other predicates.
In particular, while sleep instrumentation does somewhat increase the search space (it introduces new variables),
it leads to a significantly and \emph{qualitatively} larger solution space:
Most importantly, for some programs, $\sol{\chcTm{\P}{k}}$ is empty while $\sol{\chcSymb{\P}{k}}$ is not.

Our evaluation (\cref{sec:eval}) shows that for some programs which can be proven without reduction,
we observe a notable overhead for the instrumented version, due to the increased search space.
To minimize this overhead, we further improve upon the CHC system $\chcSymb{\P}{k}$,
by \emph{decreasing} the search space.
The improvement is quantitative, i.e., does not increase the expressivity of the approach,
but rather serves to allow CHC solvers to find a solution faster.

We observe that solutions to the symbolic-sleep encoding $\chcSymb{\P}{k}$ typically include redundant information due to symmetry:
Since the ordering of threads expressed by the thread IDs is nondeterministic,
solutions often need to make case distinctions covering all possible orderings.
Intuitively, we force the solver to prove correctness of a symmetry-equivalence class of reductions.
To illustrate this, consider again the Ashcroft invariant in \cref{fig:inc-dec-sleep-ashcroft}.
In the last two conjuncts, the invariant makes a case distinction over the ordering of thread IDs.
While this is a small example, and the Ashcroft invariant in \cref{fig:inc-dec-sleep-ashcroft} is still relatively simple,
in general (for larger programs, and larger $k$), such case distinctions can result in a factorial (in $k$) number of \emph{symmetric} conjuncts.

In order to avoid paying this additional cost,
we take advantage of the symmetry between threads.
\emph{Symmetry reductions}~\cite{clarke1998symmetry} have been widely used for parameterized systems
to reduce the search space of analyses.
The idea behind symmetry reductions is closely connected to our observations:
Instead of na\"ively enumerating all possible cases of a nondeterministically chosen order,
and recovering the same (or rather, symmetric) results for each case,
one focuses on a single fixed order.

In our case, we fix the ordering of the $k$ threads considered by the CHC predicate symbol $\Inv(g,\id{1},\pc{1},x_1,\ldots,\id{k}, \pc{k}, x_k)$
such that we always have $\id{1} < \ldots < \id{k}$.
This allows us to simplify the CHC system.
In particular, it allows us to resolve the comparisons of thread IDs in the preference test (\cref{def:pref-test}) statically.
\begin{definition}[Explicit-sleep Preference Test]
  For $i\in\{1,\ldots,k\}$, $j\in\{1,\ldots,k,\star\}$ and ${j'\in\{1,\ldots,k\}}$,
  we define the \emph{explicit-sleep preference test} as the formula
\[
  \prefCond{i}{j}{j'} :\equiv \begin{cases}
    \langle \pc{i},\pc{j} \rangle\in R & \textbf{if } i \leq j'\\
    \langle \pc{i},\pc{j} \rangle\in R  \land \lnot\langle \pc{j},\pc{i}\rangle \notin R & \textbf{otherwise}
  \end{cases}
\]
\end{definition}

Recall that $R$ is a total and transitive relation which induces a pairwise preference order (\cref{def:spo}).
We pass two parameters $j,j'$ for the second thread, in order to account for the arbitrary ordering (represented by $j'$) of the interfering thread (represented by $j=\star$)
wrt.\ to the $k$ other threads (represented by $i$).
I.e., we could have $\id{\star} < \id{1}$ (i.e., $j'=1$), or $\id{1}<\id{\star}<\id{2}$ (i.e., $j'=2$), \ldots, or $\id{k} < \id{\star}$ (i.e., $j'=k+1$).
In order to cover all cases,
we introduce one non-interference clause for each of these $k+1$ possible orderings.
Thanks to this explicit case distinction,
all comparisons between thread IDs are resolved statically.
Consequently, we eliminate the thread IDs from the CHC system completely.

Furthermore, we must take care to reorder the variables in the second line of \cref{eq:symbolic-nonint},
such that we preserve the assumption that the threads to whose variables the predicate symbol $\Inv$ is applied are in increasing order of their ID.
To this end, we define a permutation:
For $i\in\{1,\ldots,k+1\}$, $r\in\{1,\ldots,k\}$,
let $\sigma_i^r:\{1,\ldots,k\}\to\{1,\ldots,k,\star\}\setminus\{r\}$ be the bijective mapping such that
\begin{itemize}
\item for all $j_1\neq j_2$ with $\sigma_i^r(j_1)\neq\star$ and $\sigma_i^r(j_2)\neq \star$, we have $j_1 < j_2 \iff \sigma_i^r(j_1) < \sigma_i^r(j_2)$, and
\item for all $j_1\neq j_2$ with $\sigma_i^r(j_1) = \star$ and $\sigma_i^r(j_2) \neq \star$, we have $j_1 < j_2 \iff i \leq j_2$.
\end{itemize}
Intuitively, $\sigma_i^r$ describes the sequence $1,2,\ldots,k$, where we insert $\star$ at position $i$ and then delete the number $r$ from the sequence.
The index $i$ here represents the preference ordering of the interfering thread $\star$ wrt.\ the other threads $1,\ldots,k$;
index $r$ represents the thread that is replaced by $\star$.

\begin{figure}
  \footnotesize
  \begin{flalign}
    \intertext{\parbox{\linewidth}{{\em\bfseries Initial}:}}
    &\begin{aligned}[t]
      &\quad\left({\textstyle\bigwedge_{i=1}^k\mathit{pre}(g,x_1)}\right)
      \land \left( {\textstyle\bigwedge_{i=1}^k \pc{i} = \ell_\init} \right)
      \land {\textstyle\left( \bigwedge_{i=1}^k \lnot\sleep{i} \right)}\\
      &\quad\quad\quad\to \Inv(g,\pc{1},\sleep{1},x_1,\ldots,\pc{k},\sleep{k},x_k)
    \end{aligned}
    \\
    \intertext{\parbox{\linewidth}{{\em\bfseries Inductivity} (for each edge $\langle\ell,\st,\ell'\rangle \in \Delta$ and each $i\in\{1,\ldots,k\}$):}}
    &\begin{aligned}[t]
      &\quad\Inv(g,\pc{1},\sleep{1},x_1,\ldots,\pc{k},\sleep{k},x_k)
        \land \pc{i} = \ell \land \pc{i}' = \ell'
        \land \st(g,x_i,g',x_i')\\
      &\quad\quad \land \lnot\sleep{i}
        \land {\textstyle \bigwedge_{j\neq i} \left( \sleep{j}' \leftrightarrow (\sleep{j} \lor {\color{red}\prefCond{j}{i}{i}}) \land  \commCond{j}{i}{\st} \right)}
        \\
      &\quad\quad\quad \to \Inv(g',\pc{1},\sleep{1}',x_1,\ldots,\pc{i}',\sleep{i},x_i',\ldots,\pc{k},\sleep{k},x_k)
    \end{aligned}\label{eq:explicit-ind}
    \\
    \intertext{\parbox{\linewidth}{{\em\bfseries Non-Interference} (for each edge $\langle\ell,\st,\ell'\rangle\in\Delta$ {\color{red}and each $i\in\{1,\ldots,k+1\}$}):}}
    &\begin{aligned}[t]
      &\quad\Inv(g,\pc{1},\sleep{1},x_1,\ldots,\pc{k},\sleep{k},x_k)\\
      &\quad\quad \land \big({\textstyle \bigwedge_{r=1}^k \Inv(g,\pc{\color{red}\sigma_i^r(1)},\sleep{\color{red}\sigma_i^r(1)},x_{\color{red}\sigma_i^r(1)}, \ldots, \pc{\color{red}\sigma_i^r(k)},\sleep{\color{red}\sigma_i^r(k)},x_{\color{red}\sigma_i^r(k)})}\big)\\
      &\quad\quad \land \pc{\star} = \ell \land \pc{\star}' = \ell' \land \st(g,x_\star,g',x_\star')\\
      &\quad\quad \land \lnot\sleep{\star}
        \land {\textstyle \bigwedge_{j \neq i} \left( \sleep{j}' \leftrightarrow (\sleep{j} \lor {\color{red}\prefCond{j}{\star}{i}}) \land \commCond{j}{\star}{\st} \right)}
        \\
      &\quad\quad\quad \to \Inv(g',\pc{1},\sleep{1}',x_1,\ldots,\pc{k},\sleep{k}',x_k)
    \end{aligned}\label{eq:explicit-nonint}
    \\
    \intertext{\parbox{\linewidth}{{\em\bfseries Safety} (for each $i\in\{1,\ldots,k\}$ and $\ell\in\Loc$ where $\mathit{assert}(\ell)$ is defined):}}
    &\quad\Inv(g,\pc{1},\sleep{1},x_1,\ldots,\pc{k},\sleep{k},x_k) \land \pc{i} = \ell \land \lnot\mathit{assert}(\ell) \to \false
  \end{flalign}
  \caption{%
    \emph{Explicit-sleep} CHC encoding $\chcExpl{\P}{k}$ for the existence of a safe Ashcroft invariant of width $k$ for the sleep-instrumented program $\sleepInstr{\P}$.
    The encoding does not include the $\id{}$ variables.
    Further differences to the symbolic-sleep encoding are highlighted in red.
  }
  \label{fig:chc-explicit}
\end{figure}
\Cref{fig:chc-explicit} shows the resulting \emph{explicit-sleep} CHC encoding $\chcExpl{\P}{k}$, for a program $\P$ and width $k$.
Note that the explicit-sleep encoding receives the original program $\P$ as input;
the sleep instrumentation is performed as part of the encoding.
Nevertheless, we semantically connect this encoding to the sleep-instrumented program $\sleepInstr{\P}$,
in a manner analogous to \cref{lem:chc-ashcroft}.

\begin{proposition}[Explicit-Sleep Soundness]
  \label{prop:explicit-sound}
  Let $\Psi_\Inv$ be a solution to $\chcExpl{\P}{k}$.
  Then
  \[
    \forall i_1,\ldots,i_k\,.\, \id{i_1} < \ldots < \id{i_k} \to \Psi_\Inv(g,\pc{i_1},\sleep{i_1}, x_{i_1},\ldots,\pc{i_k},\sleep{i_k}, x_{i_k})
  \]
  is a safe Ashcroft invariant (of width $k$) for $\sleepInstr{\P}$.
\end{proposition}
\goodbreak

\begin{corollary}
  \label{corr:explicit-sound}
  If the explicit-sleep encoding $\chcExpl{\P}{k}$ is satisfiable, the program $\P$ is correct.
\end{corollary}
\begin{inlineproof}
  Follows from \cref{prop:explicit-sound,thm:instr-sound-red}.
\end{inlineproof}

The following proposition states that in a certain sense, the symbolic-sleep encoding and the explicit-sleep encoding are equivalent.
Consequently, the explicit-sleep encoding still encodes the existence of an Ashcroft invariant of width $k$ for the sleep-instrumented program.

\begin{proposition}[Equisatisfiability]
  \label{prop:equisat-symbolic-explicit}
  The explicit-sleep encoding $\chcExpl{\P}{k}$ is satisfiable iff the symbolic-sleep encoding $\chcSymb{\P}{k}$ is satisfiable.
\end{proposition}
\begin{inlineproof}[Proof idea]
  If $\Phi_\Inv$ is a solution for the symbolic-sleep encoding,
  then
  \begin{multline*}
    \Psi_\Inv(g,\pc{1},\sleep{1},x_1,\ldots,\pc{k},\sleep{k},x_k)\\
    :\equiv \exists \id{1},\ldots,\id{k}\,.\, \big( \id{1}<\ldots<\id{k}
    \land \Phi_\Inv(g,\id{1},\pc{1},\sleep{1},x_1,\ldots,\id{k},\pc{k},\sleep{k},x_k) \big)\end{multline*}
  is a solution for the explicit-sleep encoding.
  If $\Psi_\Inv$ is a solution for the explicit-sleep encoding,
  then
  \begin{multline*}
    \Phi_\Inv(g,\id{1},\pc{1},\sleep{1},x_1,\ldots,\id{k},\pc{k},\sleep{k},x_k) \\
    :\equiv \bigvee_{\sigma\in\mathcal{S}_k} \big( \id{\sigma(1)} < \ldots < \id{\sigma(k)}
    \land \Psi_\Inv(g,\pc{\sigma(1)},\sleep{\sigma(1)},x_{\sigma(1)},\ldots,\pc{\sigma(k)},\sleep{\sigma(k)},x_{\sigma(k)} \big)
  \end{multline*}
  is a solution for the symbolic-sleep encoding, where $\mathcal{S}_k$ denotes the set of all permutations over the set $\{1,\ldots,k\}$.
\end{inlineproof}

The factorial explosion inherent in the case distinction over all permutations of threads
is precisely the cost we seek to avoid through the explicit-sleep encoding.
Because the explicit-sleep encoding does not use variables for the thread IDs,
the search space $\search{\chcExpl{\P}{k}}$ for the explicit-sleep encoding
is a strict subset of the search space $\search{\chcSymb{\P}{k}}$ for the symbolic-sleep encoding.
The above proposition clarifies that we neither lose expressivity, nor do we gain \emph{qualitative} proof simplification,
i.e., the solution space $\sol{\chcSymb{\P}{k}}$ is empty if and only if the solution space $\sol{\chcExpl{\P}{k}}$ is empty.
Beyond that, the solution spaces are difficult to compare, as solutions range over different sets of variables.
However, symmetry reduction has been shown to be practically beneficial in many settings~\cite{clarke1998symmetry}.
And indeed, \cref{sec:eval} confirms empirically that the explicit-sleep encoding has significant practical benefit over the symbolic-sleep encoding
when using state-of-the-art CHC solvers.

\subsection{Inductive Invariants of Reduction Families}
\label{sec:red-ashcroft}

\Cref{prop:minimal} states that the sleep-instrumented program represents a family of \emph{minimal} reductions:
Every equivalence class of traces is represented by a single trace in the reduction;
if that representative is removed, the remaining set of traces is no longer a reduction.
The intention is to not burden the verification with the proof of any redundant traces.

However, this ``minimality'' refers to the family of \emph{infinite-state} programs $\sleepInstr{\P}(n)_{n\in\N}$.
When we fix a notion of \emph{finite} proofs for the parameterized program $\sleepInstr{\P}$,
we are settled with a particular expressiveness to describe this infinite family of programs.
It is not clear a priori that a certain kind of proof is expressive enough to fully benefit from the minimality of the reduction.
And in fact, if we consider Ashcroft invariants, we observe that the expressiveness Ashcroft invariants gain through sleep manifestation depends crucially on the invariants' width.
Specifically, for Ashcroft invariants of width 1, no expressivity is gained through the reduction.
\begin{proposition}[Collapse at width 1]
  Suppose there exists an Ashcroft invariant of width 1 for the sleep-instrumented program $\sleepInstr{\P}$.
  Then there also exists an Ashcroft invariant of width 1 for the original program $\P$.
\end{proposition}
Intuitively, the additional expressive power through sleep instrumentation can only be harnessed through \emph{relational} assertions,
i.e., assertions that relate the local variables (including program counter and sleep variables) of different threads.
An Ashcroft invariant of width 1 does not include such relational assertions.
It cannot even distinguish two threads.
Hence, the Ashcroft invariant can either claim that all threads are asleep (which is unsound, as there is always at least one thread awake),
or that none of the threads are asleep (i.e., there is no reduction).

By contrast, we have seen that for Ashcroft invariants of width 2 (and consequently, any higher width),
we gain expressivity through sleep instrumentation.
However, the fact that such invariants can \emph{benefit} from reduction
does not imply that they can \emph{precisely capture} the infinite family of \emph{minimal} reductions $\R(n)_{n\in\N}$
corresponding to a sleep-instrumented program $\sleepInstr{\P}$.
An Ashcroft invariant may simply capture overapproximations of the minimal reductions,
which nevertheless allow for significant (qualitative) proof simplification (as in \cref{ex:inc-dec-red-ashcroft}).
Indeed we observe:

\begin{observation}
  There exist programs for which no Ashcroft invariant of any width precisely captures the reachable configurations of the reduction.
\end{observation}
\begin{inlineproof}[Explanation]
  Consider a program $\P$ with the template

  
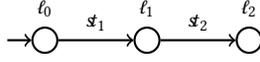
\begin{center}
  \begin{tikzpicture}[font=\footnotesize,thick]
    \node[draw,circle,label={above:$\ell_0$}] (0) {};
    \node[draw,circle,label={above:$\ell_1$},right=1cm of 0] (1) {};
    \node[draw,circle,label={above:$\ell_2$},right=1cm of 1] (2) {};
    \draw[<-] (0) -- ++(left:5mm);
    \draw[->] (0) -- node[auto] {$\st_1$} (1);
    \draw[->] (1) -- node[auto] {$\st_2$} (2);
  \end{tikzpicture}
  \end{center}

  such that all statements (of different threads) commute, except for the fact that $\inst{\st_2}{i} \notcomm \inst{\st_2}{j}$.
  An Ashcroft invariant of width $k$ would have to capture
  that,
  when the program $\P(k)$ reaches the control locations $\langle \ell_1,\ldots,\ell_1 \rangle$,
  all threads except one are in the sleep set.
  The last thread to take a step must have had the maximal thread ID, otherwise it would have been added to the sleep set earlier.
  But then, in the last step, the next enabled statement of every other thread ($\st_2$) commutes with the executed statement $\st_1$,
  thus all threads with a lower thread ID are added to the sleep set.

  Thus, if an Ashcroft invariant $\forall i_1,\ldots,i_k\,.\,\varphi$ precisely captures the reduction,
  we must have
  \begin{equation}
  \varphi \land (id_{i_1} < \ldots < \id{i_k}) \land \pc{i_1} = \ell_1 \land \ldots \land \pc{i_k} = \ell_1 \models \sleep{i_1} \land \ldots \land \sleep{i_{k-1}}
  \label{eq:capture-red-k}
  \end{equation}

  However, in the program $\P(k+1)$, we can reach a configuration with the control locations $\langle \ell_1,\ldots,\ell_1,\ell_2\rangle$,
  such that the sleep set is empty.
  In particular, if the last step is the execution of $\inst{\st_2}{k+1}$,
  this statement does not commute with the enabled statements of all other threads, thus the sleep set is emptied.
  We instantiate the Ashcroft invariant $\forall i_1,\ldots,i_k\,.\,\varphi$
  such that it considers the first $k$ threads ($\{i_1,\ldots,i_k\} = \{1,\ldots,k\}$) in increasing order of IDs ($\id{i_1} < \ldots < \id{i_k}$).
  \Cref{eq:capture-red-k} prescribes that the sleep variables $\sleep{i_1}, \ldots, \sleep{i_{k-1}}$ are true,
  when indeed  for this configuration, they are all false.
  Thus the configuration, while reachable, does not satisfy the Ashcroft invariant.
\end{inlineproof}

Note that the key obstacle to precisely capturing the reduction in the above proof was the non-commutativity of statements $\inst{\st_2}{i}$ and $\inst{\st_2}{j}$.
We observe:

\begin{observation}
  If all statements of different threads commute,
  an Ashcroft invariant of width 2 can capture a tight overapproximation of the reduction inherent in $\sleepInstr{\P}$.
\end{observation}
\begin{inlineproof}[Explanation]
  The following Ashcroft invariant precisely captures the control flow:
  \[
    \forall i,j\,.\, \id{i} < \id{j} \land \pc{j} \neq \ell_\init \to \sleep{i}
  \]
  In other words, as soon as a thread $j$ takes a step,
  all threads with smaller IDs are put to sleep and never awakened again.
  This is satisfied by all reachable configurations of $\sleepInstr{\P}$,
  and the Ashcroft invariant is \emph{precise}:
  While it may include unreachable configurations $\langle \vec\ell, s\rangle$ of $\sleepInstr{\P}$,
  such configurations are either \emph{(i)} unreachable due to data constraints (not due to the reduction),
  or \emph{(ii)}~there is a reachable configuration $\langle \vec\ell, s'\rangle$ with the same control locations and variable values,
  except that $s$ may assign more $\sleep{}$ variables to $\false$.
  The latter case is not problematic however, because $\langle \vec\ell, s\rangle$ and $\langle \vec\ell, s'\rangle$
  satisfy the same invariants over program variables,
  and all executions possible from $\langle \vec\ell, s\rangle$ are also possible from $\langle \vec\ell, s'\rangle$.
\end{inlineproof}

\section{Broader Notions of Sound Commutativity}
\label{sec:other-comm}

We have so far focused on one particular notion of commutativity (see \cref{sec:param-red}):
Executing commuting statements in either order must yield the same semantics.
The framework of commutativity theory however admits more general notions of commutativity,
from which verification can benefit.
Specifically, we extend our approach along two lines:

\paragraph{Contextual Commutativity}
The position of statements inside a program, and in an execution, provides a rich context which can benefit commutativity.
Consider for instance the statements \stsmcol{\texttt{queue[current]:=data}}{th1} (line~12) and \stsmcol{\texttt{msg:=queue[idx]}}{th2} (line~23) from the program $\P^\mathrm{notify}$
shown in \cref{fig:notify}.
These statements do not, in general, commute;
in the case that $\texttt{idx} = \texttt{current}$, executing the statements in different orders yields different semantics.

However, it is clear from the code of $\P^\mathrm{notify}$ that, in every state where these statements are enabled, it holds that $\texttt{idx} < \texttt{current}$.
In such contexts, the order of execution does indeed not matter.
Hence, we say that the statements commute in the context $\texttt{idx} < \texttt{current}$ (or, more broadly, in the context $\texttt{idx} \neq \texttt{current}$).

\paragraph{Semi-Commutativity}
Commutativity as in \cref{sec:param-red} defines a symmetric relation:
If $\inst{\st_1}{i}$ commutes with $\inst{\st_2}{j}$, then $\inst{\st_2}{j}$ commutes with $\inst{\st_1}{i}$.
The semantics of both execution orders are equal,
and consequently we can swap the statements in either direction to get an equivalent trace.
Let us entertain a non-symmetric variation.
Consider for instance the statements \stsmcol{\texttt{current:=current+1}}{th1} (line~13)
and \stsmcol{\assume{idx\,<\,current}}{th2} (line~22) from the program $\P^\mathrm{notify}$
shown in \cref{fig:notify}.
These statements do not commute.
Specifically, the execution of $\stsmcol{\inst{\assume{idx\,<\,current}\,}{$i$}}{th2} \ \ \stsmcol{\inst{\texttt{current:=current+1}\,}{$j$}}{th1}$ blocks if we have $\texttt{idx} = \texttt{current}$,
the execution of $\stsmcol{\inst{\texttt{current:=current+1}\,}{$j$}}{th1} \ \ \stsmcol{\inst{\assume{idx\,<\,current}\,}{$i$}}{th2}$ does not.
Generally, executing the increment of \texttt{current} first allows a strict superset of executions.
Thus, it is sound to eliminate a trace in which the increment happens after the \texttt{assume} statement in favor of a trace with the opposite order,
but the reverse is not true.
\medskip

Without these generalized notions of commutativity, the program $\P^\mathrm{notify}$ would not admit a simple proof.
We thus extend our approach.
\begin{definition}[Contextual Semi-Commutativity]
  Let $\varphi$ be a formula over global program variables and local variables indexed by $i$ or $j$.
  Statements $\inst{\st_1}{i}$ and $\inst{\st_2}{j}$ \emph{semi-commute in the context $\varphi$},
  denoted $\inst{\st_1}{i} \semicomm[\varphi] \inst{\st_2}{j}$,
  if for all states $s,s'$ such that $s$ satisfies $\varphi$,
  the following implication holds:
  \[
    (s,s')\in\sem{\st_1\st_2} \Rightarrow (s,s')\in\sem{\st_2\st_1}
  \]
\end{definition}

The general framework of commutativity theory is adapted accordingly.
In place of an equivalence relation, we now consider a preorder over traces (i.e., we lose symmetry).
Specifically, we say that a trace $\tau_1$ is \emph{covered by} a trace $\tau_2$
if $\tau_2$ can be derived from $\tau_1$ by a sequence of swaps of adjacent statements,
where for every swap from a trace $\tau'\ (\inst{\st_1}{i})\ (\inst{\st_2}{j})\ \tau''$ to a trace $\tau'\ (\inst{\st_2}{j})\ (\inst{\st_1}{i})\ \tau''$,
we must have $\inst{\st_1}{i} \semicomm[\varphi] \inst{\st_2}{j}$ for some $\varphi$ that always holds after the execution of the prefix $\tau'$.
A reduction language of traces $L$ is then a subset $L'$ where for every trace $\tau \in L$,
there exists some trace $\tau'\in L'$ such that $\tau$ is covered by $\tau'$.

We modify the sleep instrumentation to account for contextual semi-commutativity by redefining the commutativity test.
To this end, we assume the existence of mapping from indexed statements $\inst{\st_1}{i},\inst{\st_2}{j}$
to \emph{commutativity conditions} $\varphi_\mathrm{comm}(\inst{\st_1}{i},\inst{\st_2}{j})$,
i.e., formulae over global variables as well as local variables of threads $i$ and $j$,
such that $\inst{\st_1}{i} \semicomm[\varphi_\mathrm{comm}(\inst{\st_1\;}{\;i},\inst{\st_2\;}{\;j})] \inst{\st_2}{j}$.

\begin{definition}[Contextual Semi-Commutativity Test]
  \label{def:contextual-comm-cond}
  The \emph{contextual semi-commutativity test} $\commCond{j}{i}{\st}$ is defined as the formula
  \[
    \commCond{j}{i}{\st} :\equiv \bigvee_{\ell\in\Loc} \left( \pc{j}=\ell \land \bigwedge_{\st'\in\enabled{\ell}} \varphi_\mathrm{comm}(\inst{\st'}{j},\inst{\st}{i}) \right).
  \]
\end{definition}

At this point it is crucial that in the instrumentation $\instrument{\st}$ of a statement $\st$,
the update of the $\sleep{}$ variables, including the evaluation of the contextual semi-commutativity test,
is performed \emph{before} the original statement $\st$.
Otherwise the instrumentation would not faithfully reflect contextual semi-commutativity and might become unsound.

In the implementation of our approach (see \cref{sec:eval}),
we generate commutativity conditions $\varphi_\mathrm{comm}(\inst{\st_1}{i},\inst{\st_2}{j})$ by encoding semi-commutativity as a first-order logic formula
and applying an abduction algorithm to find sufficient conditions to guarantee it.

The modified sleep set instrumentation with contextual semi-commutativity tests still represents a reduction (\cref{thm:instr-feas-red}) and satisfies soundness (\cref{thm:instr-sound-red}) as well as conservative extension (\cref{thm:conservative-instr}).
Furthermore, the CHC encodings introduced in \cref{sec:red-proofs} can be used with the contextual semi-commutativity test in place of the commutativity test,
and remain sound.
However, the represented lexicographical reductions are not necessarily minimal~\cite{popl20:red-safety}, i.e., \cref{prop:minimal} does not hold.
This is because the covering relation is not symmetric.
There may exist traces $\tau_1,\tau_2$
such $\tau_1$ is not covered by any lexicographically smaller trace,
$\tau_2$ covers $\tau_1$,
and $\tau_2$ is only covered by itself.
Then both $\tau_1$ and $\tau_2$ appear in the lexicographical reduction, yet including $\tau_2$ would suffice.

\section{Evaluation}
\label{sec:eval}

As a proof of concept,
we have developed a tool that integrates reduction in parameterized verification.
In particular, we implemented the different CHC encodings for the existence of an Ashcroft invariant for the sleep-instrumented program $\sleepInstr{\P}$,
as discussed in \cref{sec:red-proofs}.
Our tool reads Boogie~\cite{leino:boogie} programs, generates the CHC clauses, and executes different CHC solvers to check if the CHC system is satisfiable.
In particular, we used the state-of-the-art CHC solvers \textsc{Eldarica} (\hsurl{github.com/uuverifiers/eldarica}),
\textsc{Golem} (\hsurl{verify.inf.usi.ch/Golem}) and \textsc{Z3/Spacer} (\hsurl{github.com/Z3Prover/z3}).
We evaluated the tool on a number of parameterized programs from the literature as well as custom benchmarks.
The purpose of this evaluation is to answer the following questions:
\begin{description}
\item[Q1:] Can the modular approach of \emph{(1)}~encoding reductions through sleep instrumentation
  and \emph{(2)}~subsequently verifying the resulting parameterized program
  work in practice?
\item[Q2:] Can we observe a practical benefit of the symmetry-aware \emph{explicit-sleep} CHC encoding in comparison to the default \emph{symbolic-sleep} encoding?
\end{description}
We executed the benchmarks
on a Debian 10.10 machine with a AMD Ryzen Threadripper 3970X 32-Core Processor
using the BenchExec benchmarking tool~\cite{beyer:benchexec}.
Each verification run was given a timeout of 30\,min and a memory limit of 15\,GB.
\newcommand\sat{{\color{darkgreen}sat}}%
\newcommand\unsat{{\color{red}unsat}}%
\newcommand{\repline}{\tikz[overlay] \draw (0.5em,0.8em)--(0.5em,0.4em) -- (3.5,0.4em);}
\begin{table}%
 \caption{Benchmark results. \sat\ indicates that an Ashcroft invariant (of width $k$) was found, \unsat\ indicates that a CHC solver proved that no such Ashcroft invariant exists, ``TO'' indicates a timeout ($>$ 30\,min).}%
 \label{tbl:benchmark-results}%
 \footnotesize%
 \begin{tabular}{lr|lr|lr|lr}%
    & & \multicolumn{2}{c|}{no reduction} & \multicolumn{2}{c|}{symbolic-sleep} & \multicolumn{2}{c}{explicit-sleep}\\
    Program & $k$ & status & CPU time (s) & status & CPU time (s) & status & CPU time (s) \\
    \hline
    \texttt{add-sub-nondet}            & 2 & \unsat &   20.5 & \sat   &    416.0 & \sat   &  74.5 \\
    \texttt{add-sub-positive-nondet}   & 2 & \unsat &   51.1 & \sat   & 1\,590.0 & \sat   & 144.0 \\
    \texttt{bluetooth}                 & 2 & \unsat &    6.5 & TO     &       -- & \sat   & 532.5 \\
    \texttt{equalsum-ghost}            & 2 & TO     &     -- & TO     &       -- & TO     &    -- \\
    \texttt{inc-bdec}                  & 2 & \unsat &    5.8 & \sat   &     76.3 & \sat   &  51.3 \\
    \texttt{inc-dec-eq0-locked-assert} & 2 & \sat   &   59.6 & TO     &       -- & \sat   & 726.0 \\
    \texttt{inc-dec-eq0-locked}        & 2 & \unsat &  110.0 & TO     &       -- & TO     &    -- \\
    \texttt{inc-dec-eq0}               & 2 & \unsat &    8.9 & \sat   &    112.0 & \sat   &  24.7 \\
    \texttt{inc-dec-geq0}              & 2 & \unsat &    4.3 & \sat   &      5.8 & \sat   &   5.7 \\
    \texttt{line-queue}                & 2 & TO     &     -- & TO     &       -- & TO     &    -- \\
    \texttt{lock}                      & 1 & \sat   &    4.0 & \sat   &      4.6 & \sat   &   4.7 \\
    \texttt{mutex-3}                   & 2 & \unsat &    5.2 & \sat   &      5.3 & \sat   &   4.5 \\
    \repline                           & 4 & \sat   &    8.7 & \sat   &     95.5 & TO     &    -- \\
    \texttt{mutex-4}                   & 2 & \unsat &    3.5 & \sat   &      5.6 & \sat   &   4.3 \\
    \repline                           & 5 & \sat   &   57.4 & \sat   &    723.0 & TO     &    -- \\
    \texttt{mutex-5}                   & 2 & \unsat &    4.5 & \sat   &      5.3 & \sat   &   4.0 \\
    \repline                           & 6 & \sat   &  354.0 & TO     &       -- & TO     &    -- \\
    \texttt{mutex-unbounded}           & 2 & \unsat &    4.5 & \sat   &      6.6 & \sat   &   4.2 \\
    \texttt{notify-listeners}          & 1 & TO     &     -- & TO     &       -- & \sat   & 379.0 \\
    \texttt{numbered-array}            & 2 & \sat   &    4.0 & \sat   &      5.8 & \sat   &   5.4 \\
    \texttt{thread-pooling}            & 2 & TO     &     -- & TO     &       -- & TO     &    -- \\
    \texttt{ticket}                    & 2 & \sat   &  332.0 & TO     &       -- & TO     &    --
 \end{tabular}%
\end{table}%
Our suite of 19 benchmarks is comprised of a number of variations (\texttt{inc-b?dec-*}) of the program $\P^\pm$ (see \cref{fig:inc-dec-template}),
where a variable is incremented and decremented by each thread and compared with 0;
we also included variants where a nondeterministic value is added to and subtracted from the variable (\texttt{add-sub-*}).
Several examples (namely \texttt{lock}, \texttt{ticket}, and \texttt{mutex-*}) are taken from~\cite{popl17:thread-modular};
the \texttt{mutex-*} examples correspond to the program $\P^{K\pm}$ (see \cref{fig:hierarchy-collapse}).
As more complex programs, we included the \texttt{bluetooth} example in the form presented in~\cite{popl14:proofs-count},
the example presented in \cref{sec:motivating-example} (\texttt{notify-listeners}),
the \texttt{thread-pooling} example from~\cite{popl15:proof-spaces},
a program in which each thread computes the same sum of array elements (\texttt{equalsum-ghost}),
and a custom example involving communication via queues (\texttt{line-queue}).

\Cref{tbl:benchmark-results} shows the benchmark results.
The reported CPU time encompasses both the time required to generate the CHC clauses (typically quite small)
and the time required by the fastest successful CHC solver, if any solver is successful.

Regarding \textbf{Q1},
we observe that the approach (in the explicit-sleep configuration) is able to verify 14 out of 19 benchmarks.
In particular, we successfully show correctness of non-trivial benchmarks such as \texttt{bluetooth} and \texttt{notify-listeners}.
Without reductions, these programs do not have a safe Ashcroft invariant;
a proof would require complex ghost state and/or quantified invariants.

At the same time, even for the most successful configuration (explicit-sleep),
three state-of-the-art CHC solvers are unable to solve 5 of our benchmarks.
Beyond the possibility of general improvements in CHC solving,
a possible way to improve the situation may be to guide the solvers to specifically take advantage of the reduction.
This could be beneficial in two scenarios:
First, for programs which do not have an Ashcroft invariant without reduction,
one could prevent the solver from considering solutions that ignore the instrumentation.
Second,
one could try to prevent the solver from considering solutions that use the $\sleep{}$ variables in ``exotic'' ways unsuitable to express reduction.
The second case could also reduce the overhead from the instrumentation for programs where an Ashcroft invariant exists without reduction.
As an example, consider the program \texttt{inc-dec-eq0-locked-assert}, which has an Ashcroft invariant even without reduction.
Here, the CHC solvers spend significantly longer to find a solution when the instrumentation is present.
For the program \texttt{ticket}, the solvers even time out, even though the program can be proven without reduction.

The evaluation data clearly shows the performance advantage of the explicit-sleep encoding.
With this encoding, our tool is able to verify 13 programs, compared to only 11 programs with the symbolic-sleep encoding.
Notice in particular that the complex program \texttt{notify-listeners} is only proved correct by the explicit-sleep encoding.
Furthermore, for programs solved by both the symbolic-sleep and explicit-sleep encoding,
the explicit-sleep encoding can lead to significant speedup, up to a factor of 10x in the most extreme case (\texttt{add-sub-positive-nondet}).
Despite the increased number of clauses, we do not observe any overhead for the explicit-sleep encoding.

\section{Related Work}
There is a huge body of work on verification of parameterized programs.
It is noteworthy that this paper does not put forward a new (algorithmic) framework for verifying parameterized programs,
but rather suggests a generic way of incorporating commutativity into any existing framework.
As such, we will only very briefly survey a few techniques only to justify why we chose a particular one as the framework to use for our proof of concept application.

\subsection{Parameterized Program Verification}
In {\em invisible invariants} \cite{PnueliRZ01,AronsPRXZ01},
 a candidate for an Ashcroft invariant  is constructed by first computing
the set of reachable states of the instance of the program with $k$ threads,
and then generalizing the concrete thread identifiers in the reachable states.
The candidate (a universally quantified formula with $k$ variables over thread identifiers)  is then verified using a syntactic {\em cutoff theorem}.
This approach, as well as other heuristic searchers \cite{EmmiMM10} for Ashcroft invariants,
do not have a guarantee of completeness.
Therefore they suffer from the problem that, if they fail, one does not know whether there is no proof with $k$ quantifiers
or whether the heuristic did not find it.
This is why we opted to build our reduction framework on top of thread-modular proofs~\cite{popl17:thread-modular},
which come with the guarantee of finding Ashcroft invariants when one exists (modulo incompleteness of the Horn clause solver)
or proving that no Ashcroft invariant exists.
This allows for a more principled comparison of the power of the framework in proving the original program or a lexicographical reduction of it.

In \cite{popl14:proofs-count,DBLP:conf/concur/0001KW14}, {\em counting proofs} are constructed automatically. This can be viewed as a partial solution to the problem of discovering the required ghost state automatically; partial, in the sense that only ghost counters can be discovered. Such techniques are complementary to the proposal in this paper; the simpler the proof, the more likely that a combination of this technique can succeed in discovering it automatically.

\citet{GrebenshchikovLPR12},
\citet{HojjatRSY14},
\citet{GurfinkelShoham},
and
\citet{DBLP:conf/sas/MonniauxG16}
study Horn constraints for $k$-thread-modular proofs, closely related to the framework we chose to demonstrate our approach~\cite{popl17:thread-modular}.

\subsection{Commutativity for Proof Simplification}
There has been extensive work in incorporating commutativity into verification of concurrent programs. One big cluster of such work appears under the title of partial order reduction (POR) \cite{godefroid:book,abdulla:optimal-dpor,flanagan:dpor,kahlon:monotonic-por}, and much of this work is concerned with finite-state systems or executions of bounded length.

In the context of proofs of infinite-state programs, the focus of commutativity reasoning in algorithmic verification so far has been on programs with a bounded number of threads \cite{kroening:impact,popl20:red-safety,pldi22:sound-seq,cav19:hypersafety,chu:synergize}.

\citet{DBLP:conf/fmcad/PopeeaRW14} integrate the theory of Lipton's {\em movers}~\cite{lipton:movers}
with compositional proofs in the style of Owicki and Gries,
to verify programs with a bounded number of threads.
The approach is described as a complex Horn clause system that combines compositional reasoning,
the determination of \emph{mover annotations} (i.e., commutativity checks)
and the search for reducible blocks.

In interactive proofs \cite{elmas:calculus-atomic,kragl:layered}, commutativity reasoning based on the principle of Lipton's {\em movers} has been incorporated in a way applicable to programs with a bounded number of threads as well as programs with an unbounded number of threads,
despite not explicitly using the modeling formalism of parameterized programs.
Essentially, the input program is alternatingly reduced and further abstracted.
Each abstraction step may allow more statements to commute, which enables further reduction. Since being a {\em mover} can be viewed as a local property of an atomic problem step, the size of the environment (finite vs unbounded number of threads) makes no difference in how larger atomic blocks are formed out of smaller ones by reasoning about movers,
and thus a successfully verified program is correct for any number of threads.

\citet{oopsla20:anchor} also apply mover reasoning to simplify verification of programs with an unbounded number of threads.
Data structures are annotated with \emph{synchronization specifications} that indicate mover types (i.e., semi-commutativity) of read and write accesses to the data structure.
Users specify a reduction of a concurrent program by manually instrumenting the program with \emph{yield points} indicating where interleaving with other threads may occur in the reduction.
The verifier then checks if this instrumentation is indeed sound, i.e., encodes a reduction of the program.

As discussed in \cite{popl20:red-safety}, however, the kinds of program reductions that result from Lipton's movers are not comparable with those that are produced as lexicographical reductions of (binary) commutativity relations. Besides, the locality advantages of movers disappear in the context where the goal is anything but large block reasoning: for example, a {\em lockstep} reduction.
Such reductions are by definition not local to a single thread/process.

In {\em inductive sequentialization} \cite{pldi20:ind-seq}, a vaguely similar philosophy about proof simplification is used:
Rather than reason about arbitrarily complicated executions of distributed protocol, one can reason about their equivalence to simpler ones and as such only give a proof of correctness for the simpler ones. It is important to note that the notion of equivalence employed is not the simple syntactic one (based on commutativity) used in this paper. As such, even the reasoning about such equivalences may involve the use of invariants, and other proof-type constructs. The final product is a proof of refinement between the complex and the simple protocols, and the ingredients of the proof are provided by a user.

\section{Conclusion and Future Work}

This paper proposes a methodology for incorporating commutativity reasoning into algorithmic verification of parameterized programs. We put forward the thesis that this is a worthwhile cause, because commutativity-based reductions can simplify the proofs of these programs in a precise sense: a possible substantial complexity reduction in the nature of the ghost state required for the proof. The solution was devised with an eye on practical concerns, in the sense that rather than devising a whole new algorithmic framework, one should be able to use existing frameworks for parameterized program verification with little effort.

Our investigation of this problem has led us to several new research questions that would be interesting to explore in the future. Our results from Section \ref{sec:red-ashcroft} highlight the fact that Ashcroft invariants, as a standard family of global invariants for parameterized programs, lack the expressive power to encode optimal reductions for the entire family of programs represented by the parameterized program for an arbitrary commutativity relation. It would be interesting to investigate whether this lack of expressivity is shared by other ways of giving a finitely-representable proof to a parameterized program, for instance {\em proof spaces} \cite{popl15:proof-spaces}.

Classical trace theory, which studies commutativity in a principled way, relies on a {\em finite} alphabet of program actions. For parameterized programs, one needs an infinite (indexed) alphabet of actions to model the program behaviour faithfully. Most of the work on program reductions relies on a classic result from trace theory that says ``the set of lexicographical representatives of a regular and (commutativity) closed language is regular''. The notion of regularity for indexed alphabets is less standard, and can be defined based on a number of {\em data} automata like register, nominal, or predicate automata. It will be interesting to investigate if an analogous result for these automata exists and whether it can suggest fundamentally different ways of incorporating commutativity in verification of parameterized programs.

\begin{acks}
We thank Jochen Hoenicke for his useful insights and productive discussions.
Jochen pointed out that the sequential composition of the threads of the parametrized program $\P^\pm$ (from our running example) can be accommodated via the instrumentation of the thread template with $\sleep{}$ variables, that the instrumentation yields again a parametrized program, and that this parametrized program can be handled by the proof method based on \emph{thread modularity at many levels}~\cite{popl17:thread-modular}.
\end{acks}

\bibliographystyle{ACM-Reference-Format}
\bibliography{references}

\end{document}